\documentclass[11pt]{article}
\usepackage[utf8]{inputenc}
\usepackage{amsfonts,latexsym,amsmath,amsthm,amssymb,amscd,euscript,mathrsfs,graphicx}
\usepackage{framed}
\usepackage{fullpage}
\usepackage[dvipsnames]{xcolor}
\usepackage[shortlabels]{enumitem}
\usepackage[obeyFinal,textsize=scriptsize,shadow]{todonotes}
\usepackage{tikz}
\usetikzlibrary{matrix}
\usepackage[hidelinks]{hyperref}
\usepackage{nicefrac}

\usepackage{marginnote}
\usepackage{algorithm}
\usepackage[noend]{algpseudocode}

\usepackage{accents}

\newtheorem{theorem}{Theorem}[section]
\newtheorem{proposition}[theorem]{Proposition}
\newtheorem{lemma}[theorem]{Lemma}
\newtheorem{corollary}[theorem]{Corollary}

\theoremstyle{definition}
\newtheorem{defn}[theorem]{Definition}

\theoremstyle{remark}

\newcommand{\BE}{\mathbb E}

\newcommand{\BP}{\mathbb P}

\newcommand{\BR}{\mathbb R}

\newcommand{\eps}{\varepsilon}

\newcommand{\tO}{\tilde{O}}

\title{All-Pairs Shortest Path Distances with Differential Privacy: Improved Algorithms for Bounded and Unbounded Weights} 
\author{Justin Y.\ Chen\thanks{MIT. justc@mit.edu. Supported by an NSF Graduate Research Fellowship, MathWorks Engineering Fellowship, and Simons Investigator Award.}~, Shyam Narayanan\thanks{MIT. shyamsn@mit.edu. Supported by an NSF Graduate Research Fellowship and a Simons Investigator Award.}~, Yinzhan Xu\thanks{MIT. xyzhan@mit.edu. Supported by NSF Grant CCF-1528078.}}
\date{}

\begin{document}

\maketitle

\begin{abstract}
    We revisit the problem of privately releasing the all-pairs shortest path distances of a weighted undirected graph up to low additive error, which was first studied by Sealfon~\cite{sealfon2016dpapsp}. In this paper, we improve significantly on Sealfon's results, both for arbitrary weighted graphs and for bounded-weight graphs on $n$ nodes. Specifically, we provide an approximate-DP algorithm that outputs all-pairs shortest path distances up to maximum additive error $\tilde{O}(\sqrt{n})$, and a pure-DP algorithm that outputs all pairs shortest path distances up to maximum additive error $\tilde{O}(n^{2/3})$ (where we ignore dependencies on $\varepsilon, \delta$). This improves over the previous best result of $\tilde{O}(n)$ additive error for both approximate-DP and pure-DP~\cite{sealfon2016dpapsp}, and partially resolves an open question posed by Sealfon~\cite{sealfon2016dpapsp, DPorg-open-problem-all-pairs}. We also show that if the graph is promised to have reasonably bounded weights, one can improve the error further to roughly $n^{\sqrt{2}-1+o(1)}$ in the approximate-DP setting and roughly $n^{(\sqrt{17}-3)/2 + o(1)}$ in the pure-DP setting. Previously, it was only known how to obtain $\tilde{O}(n^{1/2})$ additive error in the approximate-DP setting and $\tilde{O}(n^{2/3})$ additive error in the pure-DP setting for bounded-weight graphs~\cite{sealfon2016dpapsp}. 
\end{abstract}

\section{Introduction} \label{sec:intro}

With the advent of massive data analysis over the past few decades, one important concern has been that the data may reveal highly sensitive information about the users contributing the data. Hence, a major challenge that has arisen in data analysis is to protect the \emph{privacy} of the users. The notion of privacy that we focus on is that of \emph{differential privacy} (DP), first developed by Dwork et al.~\cite{dwork2006dp}. Differential privacy has recently emerged as the primary method of ensuring privacy both in theory and practice, and has been employed by various companies including Apple~\cite{Apple}, Google~\cite{Google}, and Microsoft~\cite{Microsoft}, as well as the US Census Bureau~\cite{Census}.

In most settings, differential privacy is utilized for algorithms that operate on a dataset $X$ of points $\{x_1, \dots, x_n\}$, where each $x_i$ represents the data of some user. Informally, an algorithm $\mathcal{A}$ operating on this dataset $X$ is said to be differentially private if the output of the algorithm does not change significantly if a single point is arbitrarily altered. The reason why this models privacy is that one can learn very little about an individual point $x_i$ from the output $\mathcal{A}(X)$, as the point $x_i$ could have been significantly altered but the output may be the same. More formally, differential privacy is defined as follows.

\begin{defn}[\cite{dwork2006dp}]
A (randomized) algorithm $\mathcal{A}$ is said to be $(\eps, \delta)$-\emph{differentially private} ($(\eps, \delta)$-DP for short) if for any two ``adjacent'' datasets $X$ and $X'$ and any subset $S$ of the output space of $\mathcal{A}$, we have
\[\BP(\mathcal{A}(X) \in S) \le e^{\eps} \cdot \BP(\mathcal{A}(X') \in S) + \delta.\]
When $\delta > 0$, this is often referred to as \emph{approximate} differential privacy (approximate-DP), as opposed to \emph{pure} differential privacy  (pure-DP) when $\delta = 0$.
\end{defn}

The definition notably requires a notion of \emph{adjacent} datasets. Normally, we say that two datasets $X = \{x_1, \dots, x_n\}, X' = \{x_1', \dots, x_n'\}$ are \emph{adjacent} if there is at most one $i$ such that $x_i \neq x_i'$, i.e., only a single point changes from $X$ to $X'$. However, in this paper, we do not focus on databases of the form $\{x_1, \dots, x_n\}$, but instead focus on private algorithms for undirected \emph{graphs}. While there have been several models designed for private graph analysis, we focus on a model first developed by Sealfon~\cite{sealfon2016dpapsp}. In this setting, we view the database as a weighted graph $H = (V, E, W)$, where $V$ represents the vertices, $E$ represents the edges, and $W \in (\BR^+)^{E}$ represents the weights on the edges. We view the underlying graph topology $G = (V, E)$, i.e., the unweighted version of the graph, as public, and we define two graphs $H, H'$ with the same underlying graph topology to be \emph{adjacent} if their edge weights differ in $\ell_1$ norm by at most $1$ when viewed as vectors in $(\BR^+)^{E}$.

While the notion of adjacency and the fact that the underlying graph is public may seem unusual, this model is natural in certain applications. In particular, Sealfon~\cite{sealfon2016dpapsp} notes that this model is well-suited for the setting of traffic data. In this case, the road network, which can be thought of as the underlying topology, is public to all, and the weights can be thought of as the amount of traffic in each individual road segment. Navigation tools such as Google Maps can estimate traffic based on GPS locations of vehicles or cell phone data, so we want to make sure that these navigation tools do not compromise the privacy of an individual vehicle or cell phone. 
Hence, the definition we use for adjacent databases is natural in this setting, as it captures the possibilities of adding or removing a single vehicle, or moving a single vehicle to a different location.

Sealfon~\cite{sealfon2016dpapsp} primarily studies the problems of private all-pairs shortest paths and private all-pairs shortest path distances on such graphs. In all-pairs shortest paths, the goal is to output the shortest path from every node $A$ to every node $B$ along a weighted graph. Indeed, a navigation app should be able to accurately estimate all such paths. In all-pairs shortest path distances, the goal is just to output the \emph{lengths} of all shortest paths.

\subsection{Prior Work and Our Results}

To describe the problem of private all-pairs shortest path distances more formally, our goal is to design an $(\eps, \delta)$-DP algorithm that operates on a weighted undirected graph $H = (V, E, W)$, and outputs a matrix $\tilde{D} \in \BR^{V \times V}$ such that $\max_{u, v \in V} |\tilde{D}(u, v) - d(u, v)|$ is small. Here, $d(u, v)$ is the true  distance between $u$ and $v$ in the graph $H$: our goal is to estimate this quantity up to small \emph{additive error} for all $u, v \in V$ in a private manner. One can also ask the similar problem of all-pairs shortest paths: here, the goal is to design an $(\eps, \delta)$-DP algorithm that operates on $H$ and outputs a path $\mathcal{P}[u, v]$ for every $u, v \in V$ such that every path is close in length to the true shortest path up to a small additive error. 

Sealfon showed that for arbitrary graphs $H = (V, E, W)$ with $|V| = n$ nodes and possibly unbounded weights, there exists an $(\eps, 0)$-DP algorithm that estimates both all-pairs shortest paths and all-pairs shortest path distances up to additive error $O(n \log n/\eps)$, i.e., every shortest path (or distance) is estimated correctly up to this additive error. He also showed that for the all-pairs shortest paths problem, one cannot obtain better than $\Omega(n)$ error, even if $\eps, \delta = \Omega(1)$ and all edge weights are bounded by some constant.

However, this does not prevent one from being able to output the \emph{distances} with smaller error. Indeed, Sealfon~\cite{sealfon2016dpapsp} showed that if all weights in $H$ are promised to be at most some value $A$, there exists an $(\eps, \delta)$-DP algorithm that estimates all-pairs shortest path distances up to additive error $\tilde{O}\left(\sqrt{n \cdot A \cdot \eps^{-1} \log \delta^{-1}}\right)$\footnote{$\tilde{O}$ hides polylog factors in $n$.}, which for small $A$ is only approximately $n^{1/2}$. Similarly, he also showed an $(\eps, 0)$-DP algorithm with additive error $\tilde{O}\left((n \cdot A)^{2/3} \eps^{-1/3}\right)$.
He also showed that for the special case where the underlying graph topology $G$ is a tree, one can estimate all-pairs shortest path distances with $(\eps, 0)$-DP up to additive error $O(\log^{2.5} n /\eps)$.
For the general case of an arbitrary graph topology and arbitrary weights, it was previously unknown how to obtain error \emph{sublinear} in $n$ for all-pairs shortest path distances. Indeed, Sealfon~\cite{sealfon2016dpapsp} poses this question as an open problem, which also later appeared on \url{differentialprivacy.org}~\cite{DPorg-open-problem-all-pairs}.

\medskip

In this paper, we improve significantly on Sealfon's results for all-pairs shortest path distances, answering this open question to a significant degree.
Specifically, we show an $(\eps, \delta)$-DP mechanism for any $0 < \eps, \delta < 1$ that, with at least $2/3$ probability, solves the all-pairs shortest path distances problem with additive error $\tilde{O}(n^{1/2} \cdot \eps^{-1} \sqrt{\log \delta^{-1}})$. In addition, we show an $(\eps, 0)$-DP mechanism for any $0 < \eps < 1$ that, with at least $2/3$ probability, solves the all-pairs shortest path distances problem with additive error $\tilde{O}(n^{2/3} \cdot \eps^{-1})$.

Furthermore, we also tackle the bounded-weight case, for which, as mentioned above, there was previously a $\tilde{O}\left(\sqrt{n \cdot A \cdot \eps^{-1} \log \delta^{-1}}\right)$-error algorithm in the approximate-DP setting and a $\tilde{O}\left((n \cdot A)^{2/3} \eps^{-1/3}\right)$-error algorithm in the pure-DP setting if all weights are in the range $(0, A]$. We provide an improved $(\eps, \delta)$-DP algorithm for this as well, which has additive error $n^{\sqrt{2}-1+o(1)} \cdot \sqrt{\log \delta^{-1}} \cdot (A^{\sqrt{2}-1} \cdot \eps^{-(2-\sqrt{2})} + \eps^{-1})$, as well as an improved $(\eps, 0)$-DP algorithm, which has additive error $n^{(\sqrt{17}-3)/2+o(1)} \cdot (A^{(\sqrt{17}-3)/2} \cdot \eps^{-(5-\sqrt{17})/2} + \eps^{-1}).$
When the bound $A$ and $\eps^{-1}, \log \delta^{-1}$ are sufficiently small (say, $n^{o(1)}$), the additive error is $n^{\sqrt{2}-1+o(1)} \le O(n^{0.4143})$ in the approximate-DP setting and $n^{(\sqrt{17}-3)/2 + o(1)} \le O(n^{0.5616})$ in the pure-DP setting. These improve over the previous best bounds of $\tilde{O}(n^{1/2})$ and $\tilde{O}(n^{2/3})$, respectively.

\subsection{Related Work}

\paragraph{Private Graph Algorithms}

The concept of differential privacy was first developed by Dwork et al.~\cite{dwork2006dp}, and differential privacy was first applied to graph data analysis a few years later by Hay et al.~\cite{hay2009graphdp}. This paper, however, focused on \emph{edge} and \emph{node} differential privacy. In these models, the underlying graph topology is supposed to be private, with adjacent data sets being represented by adding or removing a single edge or vertex, respectively. Sealfon's work~\cite{sealfon2016dpapsp} was the first to consider the model where adjacent datasets differ only in edge weights, which is the setting that we study. Since then, related models have been studied, e.g., in \cite{brunet2016privategraph, pinot2018mst, pinot2018graphclustering}.

\paragraph{All-Pairs Shortest Paths} 
All-pairs shortest paths is an important problem in graph algorithms.
Since the classic Floyd-Warshall algorithm that runs in $O(n^3)$ time for all-pairs shortest paths in weighted graphs, there have been many efforts in improving its running time (e.g.~\cite{fredman1976new, Chan10, HanTakaoka}). The current best algorithm by Williams~\cite{Williamsapsp} runs in $n^3 / \exp(\Omega(\sqrt{\log n}))$ time, which is still near cubic. 

For unweighted graphs, there is a difference between directed graphs and undirected graphs. While all-pairs shortest paths on undirected unweighted graphs can be solved in $\tilde{O}(n^\omega)$ time~\cite{seidel1995, GalilM97distance, GalilM97path} where $\omega < 2.373$~\cite{AlmanV} denotes the matrix multiplication exponent, the current best algorithm for all-pairs shortest paths on directed unweighted graphs by Zwick~\cite{Zwick02} runs in $O(n^{2.529})$ time using the best algorithm for rectangular matrix multiplication \cite{legallurr}.

While our focus is on privacy rather than runtime for all-pairs shortest path distances, we note that all of our algorithms run in polynomial time.

\paragraph{Independent work by Ghazi et al.}

Recently, an independent work by Ghazi et al.~\cite{GhaziScooped} appeared online, which also studied the problem of private all-pairs shortest path distances. They provide similar algorithms in the unbounded weight case, with roughly $n^{1/2}$ error in the $(\eps, \delta)$-DP case and roughly $n^{2/3}$ error in the $(\eps, 0)$-DP case (where we ignore factors of $\eps, \log \delta^{-1}$ for simplicity).

In addition, they also prove a lower bound of $n^{1/6}$-additive error in this case for any sufficiently small $\eps, \delta$. They also consider a multiplicative version of this problem, showing that for any integer $k \ge 2$, there is an $(\eps, \delta)$-DP (resp., $(\eps, 0)$-DP) algorithm that provides a $(2k-1)$-multiplicative approximation with a smaller $\tilde{O}(n^{(k+1)/(3k+1)})$ (resp., $\tilde{O}(n^{(k+1)/(2k+1)})$) additive error. While we do not prove results of this form, our improved approximations when the edge weights are bounded are novel even in light of their paper.

\subsection{Roadmap}

In Section \ref{sec:prelim}, we describe some notation, as well as some important preliminary results relating to differential privacy and martingales.
In Section \ref{sec:overview}, we provide an overview of both of our algorithms.
In Section \ref{sec:unbounded}, we describe and analyze our first algorithm when the weights of $H$ are possibly unbounded.
Finally, in Section \ref{sec:bounded}, we describe and analyze our second algorithm, which provides improved guarantees when the weights of $H$ are bounded by some $A$.

\section{Preliminaries} \label{sec:prelim}
Following previous work~\cite{sealfon2016dpapsp}, we only consider undirected graphs in this paper. However, we remark that our algorithm for graphs with unbounded weights also works for directed graphs. 

For any two nodes $u, v$, let $d(u, v)$ be the shortest path distance between the nodes under the edge weights and let $h(u, v)$ be the hop distance between the nodes: the shortest path distance from $u$ to $v$ under the unweighted version of $G$. If dealing with a subgraph $G' \subset G$, we use $h_{G'}(u, v)$ to denote the hop distance between $u$ and $v$ in $G'$.

Let $B_G(v, r)$ denote the ball of  radius $r$ around $v$ in $G$, i.e., $B_G(v, r) := \{u \in V: h(u, v) \le r\}$. (We may abbreviate this as $B(v, r)$ when the graph $G$ is clear.) If dealing with a subgraph $G' \subset G$, we use $B_{G'}(v, r)$ to denote $\{u \in V(G'): h_{G'}(u, v) \le r\}$.

We use $\log(x)$ to denote the natural logarithm of $x$.

Next, we describe some preliminaries related to differential privacy. First, we remark that in all of our algorithms, we implicitly assume that $\eps \le 1$ and $\delta \le \frac{1}{10}$. We next describe the Laplace mechanism. Given a vector $v \in \BR^m$, the Laplace mechanism works by replacing each coordinate $v_i$ with $v_i + Lap(T)$ for some choice of $T > 0$, where $Lap(T)$ is the Laplace distribution with PDF $\frac{1}{2T} \cdot e^{-|x|/T}$ with respect to $x$. It is well-known (see, for instance, \cite{dworkrothbook} or \cite{vadhan2017textbook}) that for any two vectors $v, v'$, the output distributions $\mathcal{D}, \mathcal{D}'$ created by applying the Laplace mechanism with parameter $T$ on $v, v'$ satisfy the following property. For any subset $S \subset \BR^m$, $\BP(\mathcal{D} \in S) \le e^{\|v-v'\|_1/T} \cdot \BP(\mathcal{D}' \in S)$. As a direct result (noted by \cite{sealfon2016dpapsp}), we have that outputting the set of weights of $H$, where each edge weight has $Lap(T)$ noise added to it, is $(1/T, 0)$-DP, since adjacent datasets differ in $\ell_1$ norm by at most $1$. Likewise, if we defined $H, H'$ to be adjacent if $\|H-H'\|_1 \le \Delta$ for some parameter $\Delta > 0$, then this mechanism would be $(\Delta/T, 0)$-DP.

In addition, one can apply the Laplace mechanism to some function $f: H \to \BR$. We say that a function $f$ has \emph{sensitivity} $\Delta$ if for any two adjacent graphs $H, H'$, $\|f(H)-f(H')\|_1 \le \Delta$. Then, applying the Laplace Mechanism by outputting $f(H) + Lap(T)$ is known to be $(\Delta/T, 0)$-DP.

Next, we note the following two theorems regarding the privacy of composing private mechanisms (see, for instance, \cite{dworkrothbook} or \cite{vadhan2017textbook}).

\begin{theorem}[Basic Composition]
    Let $\mathcal{A}_1, \dots, \mathcal{A}_k$ be mechanisms on a dataset $H$ such that each $\mathcal{A}_i$ is $(\eps_i, \delta_i)$-differentially private. Then, the mechanism $\mathcal{A}$ which concatenates the outputs of $\mathcal{A}_1, \dots, \mathcal{A}_k$ is $(\sum \eps_i, \sum \delta_i)$-differentially private.
\end{theorem}

\begin{theorem}[Strong Composition]
    Let $\mathcal{A}_1, \dots, \mathcal{A}_k$ be mechanisms on a dataset $H$ such that each $\mathcal{A}_i$ is $(\eps, \delta)$-differentially private. Then, for any $\delta' > 0$, the mechanism $\mathcal{A}$ which concatenates the outputs of $\mathcal{A}_1, \dots, \mathcal{A}_k$ is $(\sqrt{2k \log \delta^{-1}} \cdot \eps + k \eps (e^\eps - 1), k \delta + \delta')$-differentially private.
\end{theorem}

Finally, we will need Doob's optional stopping theorem, a result on martingales (we will in fact apply it to supermartingales). First, we recall the definition of a martingale and a supermartingale.

\begin{defn} \label{def:martingale}
    Let $W_0, W_1, W_2, \dots$ be a sequence of random variables, and let $Z_0, Z_1, Z_2, \dots$ be a sequence of real-valued random variables, where each $Z_i$ only depends on $W_0, \dots, W_i$. We say that $\{Z_i\}$ is a \emph{martingale} with respect to $\{W_i\}$ if for all $n \ge 1$, $\BE[Z_n|W_0, \dots, W_{n-1}] = Z_{n-1}.$ Likewise, we say that $\{Z_i\}$ is a \emph{supermartingale} with respect to $\{W_i\}$ if for all $n \ge 1$, $\BE[Z_n|W_0, \dots, W_{n-1}] \le Z_{n-1}.$
\end{defn}

Next, we recall the definition of a stopping time for a (super)martingale.

\begin{defn} \label{def:stopping_time}
    Let $\{W_i\}$ and $\{Z_i\}$ be as in Definition \ref{def:martingale}. We define a random variable $\tau$ with support on the nonnegative integers, to be a \emph{stopping time} if for all integers $t \ge 0$, the event that $\tau \le t$ is a deterministic function $W_1, \dots, W_t$. Importantly, the stopping time cannot look at future information such as $W_{t'}$ for $t' > t$ to decide when to stop.
\end{defn}

Now, we state (one version of) Doob's optional stopping theorem for supermartingales.

\begin{theorem}[Optional Stopping Theorem~\cite{probability_book}]
    Let $\{Z_i\}$ be a supermartingale with respect to $\{W_i\}$. Suppose that $\tau$ is a stopping time that is uniformly bounded with probability $1$. Then, $\BE[Z_\tau] \le \BE[Z_0]$.
\end{theorem}

\section{Overview} \label{sec:overview}
\subsection{Unbounded Weights}

Our algorithm for graphs with unbounded weights uses the ``hitting set'' idea, which was also used in Zwick's algorithm for directed unweighted all-pairs shortest paths~\cite{Zwick02}. Namely, if we fix an arbitrary shortest path $\mathcal{P}[u,v]$ between any pair of vertices $u, v$, then a random subset of $\tilde{O}(n^{\alpha})$ nodes will likely intersect all shortest paths $\mathcal{P}[u, v]$ that have at least $n^{1-\alpha}$ hops. 

We describe the high-level ideas for our approximate-DP  algorithm for graphs with unbounded weights when $\epsilon, \delta = \Theta(1)$. Let $S \subset V$ be a random subset of $\tO(\sqrt{n})$ vertices. By strong composition, releasing all $\tilde{O}(n)$ pairwise distances of nodes in $S$ plus $Lap(\tO(\sqrt{n}))$ noise is $(\epsilon / 2, \delta)$-DP (by setting appropriate factors hidden in $\tO$). We also release all the edge weights, plus $Lap(\tO(1))$ noise, so that these edge weights with noises are $(\epsilon / 2, 0)$-DP. By basic composition, the overall outputs are $(\epsilon, \delta)$-DP.

We then show how to recover the distance from $u$ to $v$ for any two nodes $u, v$. Let $s_1$ be the first vertex on $\mathcal{P}[u, v]$ that is in $S$ and let $s_2$ be the last such vertex (we assume $s_1, s_2$ exist in this section). With high probability, the number of hops on $\mathcal{P}[u, v]$ between $u$ and $s_1$ is $\tO(\sqrt{n})$ and the number of hops on $\mathcal{P}[u, v]$ between $s_2$ and $t$ is $\tO(\sqrt{n})$. We can then approximate $d(u, s_1), d(s_1, s_2), d(s_2, v)$ separately. First, we can approximate $d(s_1, s_2)$ up to $\tO(\sqrt{n})$ error, since we released $d(s_1, s_2)$ plus $Lap(\tO(\sqrt{n}))$.
Since $d(u, s_1)$ and $d(s_2, v)$ correspond to paths with $\tO(\sqrt{n})$ edges, and each outputted edge has a $Lap(\tO(1))$ noise, we can approximate them within $\tO(\sqrt{n})$ additive error. Finally, 
since $s_1, s_2$ are on the true shortest path, $d(u, v) = d(u, s_1)+d(s_1, s_2)+d(s_2, v)$, and we can approximate $d(u, v)$ by trying all possible $s_1, s_2$.

\subsection{Bounded Weights}
As a key idea in our algorithm for graphs with unbounded weights resembles the ``hitting set'' idea in Zwick's algorithm \cite{Zwick02}, it is natural to seek algorithmic tools for other variants of all-pairs shortest paths in order to get better algorithms for the bounded weight case. 

A common strategy for algorithms in undirected unweighted graphs is to use hitting set in a different way. The high-level idea there is to classify nodes as high-degree nodes and low-degree nodes (some applications also have a third class of medium-degree nodes), and to use the hitting set to hit the neighborhood of every high-degree node. This strategy was used, for instance, in additive approximate all-pairs shortest paths in undirected unweighted graphs \cite{AingworthCIM99, DorHZ00}  and additive spanners \cite{AingworthCIM99,BaswanaKMP10,Chechik13}. Our algorithm will use this common strategy as one of the key ideas, but the hitting set will be used to hit a close neighborhood of the nodes, instead of direct neighbors. 

In the following, we will sketch a $\tO(n^{3/7})$-error approximate-DP algorithm for all-pairs shortest path distances for graphs with bounded weights. 
Let  $H=(V, E, W)$ be the input graph and let $G=(V, E)$ be the public, unweighted graph. For simplicity, assume that $\epsilon, \delta = \Theta(1)$ and that the edge weights are bounded by a constant. 
Let $\mathcal{A}$ be any  $\tO(n^{1/2})$-error $(\epsilon/2, \delta / 6n^2)$-DP algorithm for all-pairs shortest path distances for graphs with bounded weights (e.g., the one in \cite{sealfon2016dpapsp}). 

First, we sample a random set $S \subset V$ of size $\tO(n^{3/7})$. By strong composition, releasing  pairwise distances of nodes in $S$ plus $Lap(\tO(n^{3/7}))$ is $(\epsilon / 2, \delta / 2)$-DP. Then we perform the following peeling procedure, as long as the graph $G$ contains a vertex $v$ where $|B(v, C n^{3/7} \log n)| \le n^{4/7}$ for a sufficiently large constant $C$.
\begin{enumerate}
    \item Sample $r \sim Expo(n^{3/7})$, i.e., $r$ has an exponential distribution with mean $n^{3/7}$;
    \item Let $H'$ be a subgraph of $H$, where the node set of $H'$ is $B(v, r + 1)$ and the edge set of $H'$ is the set of all edges that are adjacent to $B(v, r)$ (note that with high probability, $r+1 \le C n^{3/7} \log n$, so $|V(H')| \le n^{4/7}$ and note that the vertex/edge structure of $H'$ can be computed from $G$ without access to the edge weights); 
    \item Run $\mathcal{A}$ on $H'$ and release its outputs;
    \item Remove all vertices in $B(v, r)$ and their adjacent edges from $H$. 
\end{enumerate}
Note that the output of $\mathcal{A}$ for each $H'$ is $(\epsilon/2, \delta / 6n^2)$-DP, and the edge set of all $H'$ are disjoint, so intuitively we should expect the combined outputs of all calls to $\mathcal{A}$ to have a good privacy guarantee. In fact, the combined outputs are $(\epsilon/2, \delta / 2)$-DP (as we will show in Section~\ref{subsec:bounded_privacy}). Overall, what we release is $(\epsilon, \delta)$-DP. 

Now we show how to recover the pairwise distances up to additive error $\tO(n^{3/7})$. For any $u, v \in V$, consider an arbitrary shortest path $\mathcal{P}$ from $u$ to $v$. Let $p_i$ be the first vertex on this path such that $S \cap B(p_i, C n^{3/7} \log n) \ne \emptyset$, and let $p_j$ be the last such vertex (we assume $p_i, p_j$ exist in this section). Let $s_1 \in S \cap B(p_i, C n^{3/7} \log n)$ and let $s_2 \in S \cap B(p_j, C n^{3/7} \log n)$. Note that $|d(u, v) - (d(u, p_i) + d(s_1, s_2) + d(p_j, v))| \le d(p_i, s_1) + d(p_j, s_2) = \tO(n^{3/7})$, and we have an $\tO(n^{3/7})$ error approximation of $d(s_1, s_2)$ since $s_1,s_2 \in S$. It remains to consider $d(u, p_i)$ (and $d(p_j, v)$ follows by symmetry). 

With high probability, the number of hops from $u$ to $p_i$ on $\mathcal{P}$ is $\tO(n^{4/7})$ since among the first $\tO(n^{4/7})$ vertices on $\mathcal{P}$, one of them belongs to $S$ with high probability. Also, with high probability, none of the vertices $x$ on $\mathcal{P}$ before $p_i$  can have $|B(x, Cn^{3/7} \log n)| > n^{4/7}$ (since otherwise $B(x, Cn^{3/7} \log n)$ will likely intersect $S$ and we would choose $x$ to be $p_i$), and consequently all edges between $u$ and $p_i$ on $\mathcal{P}$ will be removed in the peeling procedure. 

Recall in the peeling procedure, each time we pick $v$ and sample $r$ from $Expo(n^{3/7})$.  If a node $x$ on $\mathcal{P}$ before $p_i$ is  in $B(v, r)$, then we would expect to see all nodes $y$ around $x$ up to hop distance $n^{3/7}$ to belong to $B(v, r)$ as well. Therefore, intuitively, we should be able to decompose the subpath of $\mathcal{P}$ from $u$ to $p_i$ to $\tO(n^{4/7} / n^{3/7}) = \tO(n^{1/7})$ subpaths, and each subpath is contained completely inside one $H'$.
Since each $H'$ has $O(n^{4/7})$ nodes as mentioned, algorithm $\mathcal{A}$ can provide an $\tO(\sqrt{n^{4/7}}) = \tO(n^{2/7})$ additive approximation for all-pair shortest path distances inside each $G'$. Therefore, the overall additive error of $d(u, p_1)$ is bounded by $\tO(n^{1/7} \cdot n^{2/7}) = \tO(n^{3/7})$, as desired. 

In Section~\ref{sec:bounded}, we will combine a more refined version of the above intuition with recursion to 
get the approximate-DP algorithm with $n^{\sqrt{2}-1 + o(1)}$ additive error, as well as a pure-DP algorithm with $n^{(\sqrt{17}-3)/2 + o(1))}$ additive error (when $\eps, \delta = \Omega(1)$ and all edge weights are bounded by some constant).

\section{Algorithms for Unbounded Weights} \label{sec:unbounded}

In this section, we describe pure and approximate-DP algorithms for the all-pairs shortest path distances problem with for graphs with unbounded, positive edge weights.

\begin{algorithm}
\caption{}
\label{alg:unbounded}
\begin{algorithmic}[1]
\Procedure{UnboundedWeights}{$H = (V, E, W), \eps, \delta$}
\State For pure-DP, set $L \leftarrow O(n^{1/3}\log n)$ and $T \leftarrow O(n^{2/3}\eps^{-1} \log^2 n)$
\State For approximate-DP, set $L \leftarrow O(n^{1/2} \log n)$ and $T \leftarrow O\left(n^{1/2}\eps^{-1} \log n \sqrt{\log \frac{1}{\delta}}\right)$
\State Pick a random set of vertices $S \subset V$ of size $L$
\State For each pair $u, v \in S$, compute $d(u, v)$ and release the distance plus $Lap(T)$ noise
\State For each edge $e \in E$, release the edge weight $w_e$ plus $Lap(2/\eps)$ noise
\EndProcedure
\end{algorithmic}
\end{algorithm}

\begin{theorem}
\label{thm:unboundedpuredp}
In the pure-DP setting, Algorithm~\ref{alg:unbounded} is $(\eps, 0)$-DP.
\end{theorem}

\begin{proof}
We will separately analyze the privacy loss due to releasing shortest path distances within $S$ and to releasing all edge weights.
Consider any pair $u, v \in S$.
The distance between $u$ and $v$ has sensitivity 1. Changing the vector of edge weights by at most 1 in $\ell_1$ distance can alter the length of any path by at most 1, and so the length of the shortest path between any two vertices can also change by at most 1.
Therefore, via the Laplace mechanism, the noisy distance between $u$ and $v$ released by the algorithm is $(\frac{1}{T}, 0)$-DP.
Via basic composition, the set of all $L^2$ such distances is $(\frac{L^2}{T}, 0)$-DP.
Plugging in the values of $L$ and $T$ for pure-DP with appropriate constant factors, releasing the distances for pairs within $S$ is $(\eps/2, 0)$-DP.

As the set of edge weights is exactly the object that defines neighboring datasets, the sensitivity of the set of edge weights is 1.
Therefore, by the Laplace mechanism, releasing the noisy edge weights is $(\eps/2, 0)$-DP.
By basic composition, the entire algorithm is $(\eps, 0)$-DP, as required.
\end{proof}

\begin{theorem}
In the approximate-DP setting, Algorithm~\ref{alg:unbounded} is $(\eps, \delta)$-DP.
\end{theorem}

\begin{proof}
This result follows from the argument for pure-DP with the one change of using strong composition rather than basic composition for the distances between pairs in $S$. Using strong composition, releasing the set of $L^2$ distances is $\left(\frac{L\sqrt{\log(1/\delta)}}{T}, \delta \right)$-DP.
Plugging in the values of $L$ and $T$ for the approximate-DP setting with appropriate constants, the release is $O(\eps/2, \delta)$-DP.
Via basic composition with the $(\eps/2, 0)$-DP release of noisy edge weights, in total the algorithm is $O(\eps, \delta)$-DP, as required.
\end{proof}

\begin{theorem}
\label{thm:unboundederror}
Given the unweighted structure of the graph $G = (V,E)$ and the output of Algorithm~\ref{alg:unbounded} on the weighted version of the graph $H=(V, E, W)$, all-pairs shortest path distances can be computed within additive error $O(T \log n)$ with constant success probability.
\end{theorem}

\begin{corollary}
In the pure-DP setting, all pairs shortest path distances can be computed up to  additive error $O(n^{2/3}\eps^{-1} \log^3 n)$.
\end{corollary}

\begin{corollary}
In the approximate-DP setting, all pairs shortest path distances can be computed up to  additive error $O(n^{1/2} \eps^{-1} \log^2 n \sqrt{\log{1/\delta}})$.
\end{corollary}

\begin{proof}[Proof of Theorem~\ref{thm:unboundederror}]
Fix any nodes $u, v$. To compute the approximate distance between $u$ and $v$, we will take the minimum of two separate estimators. 
Let $R = 10 n\log n / L$ (so $R = O(n^{2/3})$ in the pure-DP case and $R = O(\sqrt{n})$ in the approximate-DP case).
First compute the shortest path of at most $R$ hops between $u$ and $v$ using the noisy edge weights output by Algorithm~\ref{alg:unbounded}. Call its distance $\hat{d}_{short}(u,v)$ (if no $R$ hop path exists, set $\hat{d}_{short}(u, v) = \infty$).

Let $S_u = \{w \in S : h(u, w) \leq R\}$ be the set of nodes in $S$ within a hop radius of $R$ from $u$. Let $\hat{d}_S(i, j)$ be the approximate distance between two nodes $i, j \in S$ as output by the algorithm. Then, let
\[
\hat{d}_{long} = \min_{i \in S_u, j \in S_v} \left( \hat{d}_{short}(u, i) + \hat{d}_S(i,j) + \hat{d}_{short}(j, v)\right).
\]
We will argue that for all pairs $u, v \in V$, the estimate $\min\{\hat{d}_{short}(u, v), \hat{d}_{long}(u, v)\}$ is within $O(T \log n)$ of the true shortest path distance $d(u, v)$.

Consider any $R$ hop path $P$. With high probability, no edge weight has noise of magnitude greater than $O(\log n /\eps)$. Under the Laplace distribution, the probability that the sampled $Lap(2/\eps)$ noise exceeds $6\log n / \eps$ is at most $e^{-(6 \log n/\eps)/(2/\eps)} = e^{-3\log n} = n^{-3}$. Union bounding over all edges, with probability $1 - \frac{1}{n}$, all edge weights have noise below this threshold.
Therefore, the estimated length of $P$ using the noisy edge weights has error at most $O(R \log n / \eps) = O(T)$.

Consider any path that is the concatenation of a path from $u$ to $i$, a shortest path from $i$ to $j$, and a path from $j$ to $v$ for some $i \in S_u$ and $j \in S_j$. Call any such path a \emph{composite} path. By the reasoning above, using the noisy edge weights, the distances for the first and last subpaths are well-approximated up to error $O(T)$. With high probability, all of the noisy $S$ to $S$ distances released by the algorithm have error at most $O(T \log L^2) = O(T \log n)$ once again by the Laplacian tail bounds.
Therefore, for all such composite paths, the total approximation error is $O(T \log n)$.
To finish the analysis, we will proceed by cases.

\paragraph{Case 1:} There exists a shortest path $P$ from $u$ to $v$ of hop length at most $R$.
As we can well-approximate short paths and composite paths with high probability,
\[
    |\hat{d}_{short}(u, v) - d(u, v)| = O(T)
\]
and
\[
    \hat{d}_{long}(u, v) \geq d(u, v) - O(T\log n).
\]
Therefore,
\[
    |\min\{\hat{d}_{short}(u, v), \hat{d}_{long}(u, v)\} - d(u, v)| = O(T\log n).
\]

\paragraph{Case 2:} There exists a shortest path $P$ from $u$ to $v$ of hop length greater than $R$.
Then, with high probability, $P$ is a composite path.
In order for $P$ not to be a composite path, then there must be no elements in $S$ within either the first $R$ nodes of $P$ or the last $R$ nodes of $P$.
As $S$ is generated randomly, the probability that there are no elements in $S$ within the first $R$ nodes of $P$ is
\[
(1 - L/n)^R = (1 - L/n)^{10 n\log n / L} \leq e^{-10 \log n} = n^{-10}.
\]
As this holds symmetrically for the probability of no elements in $S$ appearing in the last $R$ nodes of $P$, with high probability, $P$ is a composite path, and
\[
    |\hat{d}_{long}(u, v) - d(u, v)| = O(T \log n).
\]
Note also that the short path estimate will not be too small:
\[
    \hat{d}_{short}(u, v) \geq d(u, v) - O(T).
\]
Therefore,
\[
    |\min\{\hat{d}_{short}(u, v), \hat{d}_{long}(u, v)\} - d(u, v)| = O(T\log n),
\]
completing the proof.
\end{proof}

\section{Algorithms for Bounded Weights} \label{sec:bounded}
In this section, we describe pure and approximate-DP algorithms for the all-pairs shortest path distances problem for graphs with edge weights bounded in the range $(0, A]$ for some bound $A$. When $A$ is sufficiently small, the error for each distance is at most approximately $n^{\sqrt{2}-1}$ in the approximate-DP setting, which improves over the approximately $\sqrt{n}$ upper bound for the unbounded weights case.

\subsection{Peeling Procedure for unweighted graphs} \label{subsec:peeling}

Before describing our overall algorithm, we first develop a randomized peeling procedure on unweighted graphs, which will have crucial properties that help us in our private all-pairs shortest path distances algorithm.

Fix $G = (V, E)$ to be an arbitrary \emph{unweighted} graph, and fix an arbitrary path $\mathcal{P} = (p_0, p_1, \dots, p_Q)$ of vertices on the graph. Initially, set every vertex $v \in V$ to have color $0$. Now, consider the following iterative procedure. At time step $1$, pick a vertex $v_1$, and sample $r_1 \sim Expo(R)$, i.e., $r_1$ has an exponential distribution with mean $R$. Recall that $B(v_1, r_1)$ represents the set of vertices within distance $r_1$ of $v_1$ in $G$ (note that this always includes $v_1$). Color all vertices in $B(v_1, r_1)$ with color $1$, and then remove all vertices in $B(v_1, r_1)$ from $G$. However, we keep the knowledge of which vertices were removed and their color. At time step $2$, in the remaining graph, pick a new vertex $v_2$, sample $r_2 \sim Expo(R)$, construct $B(v_2, r_2)$ similarly (but on the graph with $B(v_1, r_1)$ removed), color these vertices with color $2$, and then remove them. Keep repeating this process until a designated stopping time $\tau$ which only depends on which vertices have been removed so far.

Let $X_0, X_1, X_2, \dots$ be a sequence of random variables, where $X_t$ represents the random variable that counts the number of pairs $(p_i, p_{i+1})$ of consecutive vertices on the path $\mathcal{P}$ with different colors after time step $t$. (We also use $X_0$ to denote this random variable before time step $1$). Let $X = X_\tau$. Then, the following holds.

\begin{lemma} \label{lem:martingale_peeling}
    Suppose that the choice of each $v_t$ only depends on the original graph $G$ and which vertices in $G$ have been removed before time step $t$. (Importantly, $v_t$ is not allowed to depend on the choice of $\mathcal{P}$.) Then, $\BE[X] \le \frac{4Q}{R}$.
\end{lemma}

\begin{proof}
    Let $W_t$ represent the sequence of colors of the vertex set $V$ after time step $t$ (where $W_0$ is the initial sequence of all $0$'s).
    Also, let $Y_t$ represent the number of vertices $p_i$ on $\mathcal{P}$ with a nonzero color after time step $t$ (and $Y_0 = 0$). Note that $X_t, Y_t$ are deterministic functions of $W_t$, and $W_t$ is only dependent on $W_0, \dots, W_{t-1}$ and the random variable $r_t$.
    
    Consider the time right before time step $t$, when $W_{t-1}, X_{t-1}, Y_{t-1}$ are fixed. Suppose the algorithm picks a vertex $v_t$ of color $0$ (which may or may not be on the path $\mathcal{P}$), and consider any arbitrary vertex $p_i$ on $\mathcal{P}$. If $p_i$ has a nonzero color before time step $t$, then if $p_i, p_{i+1}$ share the same color, their color will not change, so they continue to always share the same color. In addition, if $p_i$ has nonzero color and $p_i, p_{i+1}$ do not share the same color before time step $t$, they will continue to not share the same color.
    
    let $h_i := h(p_i, v_t)$ represent the distance from $v_t$ to $p_i$ in the graph with $B(v_1, r_1), \ldots, B(v_{t-1}, r_{t-1})$ removed. If $p_i$ has color $0$ before time step $t$, the probability that $p_i$ will be colored with color $t$ is clearly $\BP(h_i \le Expo(R)) = e^{-h_i/R}$. Note that $|h_i-h_{i+1}| \le 1$ for all $i$. Suppose $p_i$ and $p_{i+1}$ both have color $0$ before time step $t$. If $h_i  = h_{i+1} - 1$, then the only possibility that $p_i$ and $p_{i+1}$ share different colors right after time step $t$ is when $p_i$ is colored with color $t$ and $p_{i+1}$ is not, which happens with probability at most $e^{-h_i/R}-e^{-(h_i+1)/R} = e^{-h_i/R} \cdot (1-e^{-1/R}) \le \frac{2}{R} \cdot e^{-h_i/R}$. Similarly, if $h_i = h_{i+1} + 1$, the probability that $p_i$ and $p_{i+1}$ share different colors is at most $e^{-(h_i-1)/R}-e^{-h_i/R} = e^{-h_i/R} \cdot (e^{1/R}-1) \le \frac{2}{R} \cdot e^{-h_i/R}$. If $h_i = h_{i+1}$, $p_i$ and $p_{i+1}$ cannot have different colors.

    
    Note that $\BE[Y_t-Y_{t-1}|W_{t-1}]$ exactly equals the expected number of vertices on $\mathcal{P}$ that are colored with color $t$. By the above paragraph and linearity of expectation, we therefore have that 
\[\BE[Y_t-Y_{t-1}|W_{t-1}] = \sum_{\substack{i: p_i \in \mathcal{P} \\i \text{ has color } 0 \text{ before time } t} } e^{-h_i/R}.\]
    Next, note that $\BE[X_t-X_{t-1}|W_{t-1}]$ is at most the expected number of pairs $(p_i, p_{i+1})$ which previously had the same color (which would have to be the $0$ color) before time step $t$ but have different colors after time step $t$. By the above paragraph and linearity of expectation, we therefore have 
\[\BE[X_t-X_{t-1}|W_{t-1}] \le \sum_{\substack{i: p_i \in \mathcal{P} \\i \text{ has color } 0 \text{ before time } t} } \left(\frac{2}{R} \cdot e^{-h_i/R}\right) = \frac{2}{R} \cdot \sum_{\substack{i: p_i \in \mathcal{P} \\i \text{ has color } 0 \text{ before time } t} } e^{-h_i/R}.\]
    Therefore, $\BE[X_t-X_{t-1}|W_{t-1}] \le \frac{2}{R} \cdot \BE[Y_{t}-Y_{t-1}|W_{t-1}]$.
    
    Now, let us define $Z_t = X_t - \frac{2}{R} \cdot Y_t$. Then, we have that $\BE[Z_t-Z_{t-1}|W_{t-1}] \le 0$, which means that $\BE[Z_t|W_{t-1}] \le Z_{t-1}$. So, we have that $Z_t$ is a \emph{supermartingale} with respect to $W_0, W_1, \dots$. Note that the stopping time $\tau$ only depends on the $W_t$'s that have been seen so far, and $\tau \le |V|$ almost surely because each time step we remove at least one vertex. Therefore, we can use the optional stopping theorem to say that $\BE[Z_\tau] = \BE[Z] \le 0$, so $\BE[X_\tau] \le \frac{2}{R} \BE[Y_\tau]$. But of course, $Y_\tau \le Q+1$ since there are at most $Q+1$ vertices in $\mathcal{P}$, so $\BE[X] \le \frac{2}{R} \cdot (Q+1) \le \frac{4 Q}{R}$.
\end{proof}


We again consider the unweighted graph $G = (V, E)$, and consider the following peeling algorithm similar to our coloring procedure. Let $G'$ represent the remaining graph before time step $t$. At time step $t$, pick an arbitrary vertex $v_t \in G'$, such that the ball of radius $100 R \log n$ in $G'$ around $v_t$ has size at most $T := n^b$ for some constant $b$. Sample $r_t \sim Expo(R)$, and remove $B^{(t)}$, which is the ball of radius $r_t$ in $G'$ around $v_t$, from $G'$. We repeat this process until there does not exist any vertex $v \in G'$ remaining such that the ball of radius $100 R \log n$ around $v_t$ has size at most $T$.

Now, fix a uniformly random set $S \subset V$ of size $L = 100 n^{1-b} \log n$, independently of the peeling process. In this case, we claim the following propositions.

\begin{proposition} \label{prop:peeling_hitting_set}
    Suppose that $\mathcal{P} = (p_0, \dots, p_Q)$ is a path in $G$, and suppose that $Q \ge 2T$. Then, with probability at least $1- n^{-100}$, there exists an index $i$ such that $i \le T$, $p_i$ has distance at most $100 R \log n$ from some vertex $s_1 \in S$, and every vertex $p_{i'} \in \mathcal{P}$ for $i' < i$ has been peeled off in some $B^{(t)}$. 
    
    By a symmetric argument, with probability at least $1- n^{-100}$, there exists an index $j$ such that $j \le Q-T$, $p_j$ has distance at most $100 R \log n$ from some vertex $s_2 \in S$, and every vertex $p_{j'} \in \mathcal{P}$ for $j' > j$ has been peeled off in some $B^{(t)}$.
\end{proposition}

\begin{proof}
    We only focus on the first claim, as the second claim is symmetric. First, suppose there exists an index $i \le T$ such that $p_i$ was not peeled off by any $B^{(t)}$. If such an index exists, pick $p_i$ to have the minimum such index $i$. Then, because $p_i$ was never peeled off, the ball of radius $100 R \log n$ around $p_i$ in the final graph $G' \subset G$ contains at least $T = n^b$ vertices, so the ball of radius $100 R \log n$ around $p_i$ in $G$ also must contain at least $T = n^b$ vertices. So, with failure probability at most $(1-L/n)^T \le e^{-L \cdot T/n} \le n^{-100}$, there is a vertex $s_1 \in S$ in the ball of radius $100 R \log n$ around $p_i$, which means $s_1$ has distance at most $100 R \log n$ away from $p_i$. Also, since we chose the smallest possible $p_i$, all of $p_0, \dots, p_{i-1}$ have been peeled off.
    
    Alternatively, if no such $i \le T$ exists, then since $Q \ge 2T \ge T$, there exists a vertex among $p_0, \dots, p_T$ which is in the random set $S$ with failure probability at most $(1-L/n)^T \le n^{-100}$. So, with probability at least $1 - n^{-100}$, there exists some $p_i = s_1 \in S$, and all of $p_0, \dots, p_{i-1}$ (as well as $p_i, \dots, p_T$) have been peeled off.
\end{proof}


\begin{proposition} \label{prop:peeling_hitting_set_2}
    Suppose that $\mathcal{P} = (p_0, \dots, p_Q)$ is a path in $G$, and suppose that $Q < 2T$. Then, with probability at least $1-2n^{-100}$, either every vertex on the path $\mathcal{P}$ has been peeled off in some $B^{(t)}$, or there exist indices $i \le j$ such that both of the following occur:
    \begin{itemize}
        \item $p_i$ has distance at most $100 R \log n$ from some $s_1 \in S$ and every vertex $p_{i'}$ for $i' < i$ has been peeled off.
        \item $p_j$ has distance at most $100 R \log n$ from some $s_2 \in S$ and every vertex $p_{j'}$ for $j' > j$ has been peeled off.
    \end{itemize}
\end{proposition}

\begin{proof}
    Assume that some vertex in $\mathcal{P}$ is not peeled off, in which case our goal is to verify the two bullets in the proposition statement.
    By symmetry, we can WLOG focus on the first one. Let $i$ be the smallest index such that $p_i$ is not peeled off by some $B^{(t)}$. (We would similarly let $j$ be the largest index such that $p_j$ is not peeled off by some $B^{(t)}$, so $i \le j$ by default.) Then, as in Proposition \ref{prop:peeling_hitting_set}, we have that the ball of radius $100 R \log n$ around $p_i$ in $G$ contains at least $T = n^b$ vertices, so with probability at least $1-n^{-100}$, there is a vertex $s_1 \in S$ in the ball of radius $100 R \log n$ around $p_i$, which means $s_1$ has distance at most $100 R \log n$ from $p_i$. Also, since we chose the smallest possible $p_i$, all of $p_0, \dots, p_{i-1}$ have been peeled off.
\end{proof}

\subsection{Algorithm} \label{subsec:bounded_algo}

The algorithm $\Call{BoundedWeights}{}$ is written out as pseudocode as Algorithm \ref{alg:smallweights}. We describe the algorithm in words as follows.

Let $H = (V, E, W)$. Let $G = (V, E)$ represent the unweighted graph topology, and let  $R, T, L$ be some parameters such that $L = \tilde{O}(n/T)$. We will set these parameters in Subsection \ref{subsec:bounded_accuracy}. We start by considering a copy $G'$ of the unweighted graph $G$, and perform the peeling algorithm similar to Subsection \ref{subsec:peeling}. Specifically, at each time step $t$, we look for a vertex $v_t \in G'$ such that the ball of radius $100 R \log n$ around $v_t$ has size at most $T$, and if so, peel off the ball $B^{(t)} = B_{G'}(v_t, r_t)$ from $G'$, where $r_t \sim Expo(T)$. (So, $G'$ keeps shrinking). We iteratively repeat this until we can no longer find a desired vertex $v_t$.

Next, for each ball $B^{(t)}$, we consider it as a subset of the vertices $V$, and recursively apply the algorithm $\Call{BoundedWeights}{}$ on each ball $B^{(t)}$, but with a slightly stronger $\left(\frac{\eps}{(\log n)^{O(1)}}, \frac{\delta}{\text{poly(n)}}\right)$-DP. This allows us to get reasonably accurate shortest paths for any vertices $u, v$ that were peeled off by the same ball. (In reality, for each $B^{(t)}$, we take the median of $O(\log n)$ copies of $\Call{BoundedWeights}{}$ to amplify the success probability.)

Next, we create a random subset $S \subset V$ of size $L$, which we call the hitting set. Now, we create an auxiliary graph $\tilde{H}$, which will be a noisy version of $H$ with some additional edges. First, each edge in $E$ is added to $\tilde{H}$, but Laplace noise is added to the weight of each edge $(u, v) \in E$. We also add an edge between each pair $u, v$ in the hitting set $S$, with weight equal to the shortest path in $H$ plus some Laplace noise (See the pseudocode for the parameters of the Laplace noises). Finally, we add an edge between all $u, v$ that are in the same peeled ball, based on the recursive call that approximated its shortest path distance. Overall, we have three types of edges, which we think of as red, blue, and green, respectively. The final algorithm, for each pair $u, v \in V$, finds the shortest path distance from $u$ to $v$ that uses at most $\tilde{O}(R + T/R)$ red edges, at most $1$ blue edge, and at most $O(T/R)$ green edges.

Finally, to amplify the success probability, we repeat this whole algorithm $O(\log n)$ times (Line 5) and take the median answer for each shortest path distance (Line 27).

\begin{algorithm}
\caption{}
\label{alg:smallweights}
\begin{algorithmic}[1]
\Procedure{BoundedWeights}{$H = (V, E, W), \eps, \delta$} \Comment{Give an $(\eps, \delta)$-DP algorithm that computes all-pairs shortest path distances on the weighted graph $H$ with edge weights $W$ bounded by $A$}
\State Set $K \leftarrow 100 \log n$, $R, T$ to be determined, and $L \leftarrow 100 \log n \cdot \frac{n}{T}$ \Comment{We will set $T = \frac{n}{(A \eps n)^{\sqrt{2}-1}}, R = \frac{n}{(A \eps n)^{2-\sqrt{2}}}$ in the approx-DP setting, and $T = \frac{n}{(A \eps n)^{(\sqrt{17}-3)/4}}, R = \frac{n}{(A \eps n)^{(5-\sqrt{17})/2}}$}
\State Define $G = (V, E)$ to be the underlying unweighted graph topology of $H$
\State Initialize matrices $\tilde{D}^{(1)}, \dots, \tilde{D}^{(K)}$
\For{$k = 1$ to $K$} \Comment{Will repeat the overall procedure $O(\log n)$ independent times}
    \State $G' \leftarrow G$ \Comment{Create new copy of unweighted graph $G$}
    \State $t \leftarrow 0$ \Comment{Counter for counting the sets we peel off}
    
    \While{$\exists v_t \in G'$ such that $|B_{G'}(v_t, 100 R \log n)| \le T$} \Comment{If there exist multiple choices for $v_t$, pick one arbitrarily}
        \State Sample $r_t \sim Expo(R)$
        \State Let $B^{(t)} := B_{G'}(v_t, r_t)$ be the set of vertices that are within distance $r_t$ of $v_t$ in $G'$
        \State $G' \leftarrow G' \backslash B^{(t)}$ \Comment{Peel off ball around $v_t$}
        \State $t \leftarrow t+1$ \Comment{Increment counter}
    \EndWhile

    \For{all sets $B^{(t)}$} \Comment{Recursively apply procedure on smaller balls}
        \For{$\ell = 1$ to $K$}
            \State Compute the matrix $M^{(t, \ell)} = \Call{BoundedWeights}{H[B^{(t)}], \frac{\eps}{3K^2}, \frac{\delta}{3K^2}}$ \Comment{For each iteration $\ell$ the matrix is re-computed with fresh randomness. Also, replace $\frac{\delta}{3K^2}$ with $0$ in pure-DP case.}
        \EndFor
        Let $M^{(t)}$ be the entrywise median of $M^{(t, \ell)}$ over $1 \le \ell \le K$
    \EndFor
    
    \State Pick a uniformly random set $S \subset V$ of size $L$ \Comment{Construct hitting set}
    \State Create a weighted multigraph $\tilde{H}$ with vertex set $V$ and initially no edges
    
    \For{each pair $(u, v) \in E$} 
        \State add a red edge $(u, v) \in E$ to $\tilde{H}$ of weight $w_{u, v} + Lap(3K/\eps)$
    \EndFor
    
    \For{all pairs $u, v \in S$}
        \State add a blue edge $(u, v)$ to $\tilde{H}$ of weight $d(u, v) + Lap(10 K L \sqrt{\log \frac{3K}{\delta}}/\eps)$ \Comment{Use $Lap(10 K L^2 /\eps)$ if pure-DP is desired}
    \EndFor
    
    \For{each $B^{(t)}$}
        \For{all pairs $u, v \in B^{(t)}$}
            \State add a green edge $(u, v)$ to $\tilde{H}$ of weight $M^{(t)}[u, v]$
        \EndFor
    \EndFor
    
    \For{all pairs $u, v \in V$}
        \State Let $\tilde{D}^{(k)}[u, v]$ be the shortest path distance on $\tilde{H}$ from $u$ to $v$ that uses at most $100 R \log n+\frac{100 T}{R}$ red edges, at most $1$ blue edge, and at most $\frac{100T}{R}$ green edges.
    \EndFor

\EndFor
\State \textbf{Return} $\tilde{D}$, the entrywise median of $\tilde{D}^{(1)}, \dots, \tilde{D}^{(K)}$ \Comment{Entrywise median used to improve success probability}
\EndProcedure
\end{algorithmic}
\end{algorithm}

\subsection{Privacy Analysis} \label{subsec:bounded_privacy}
As the final distance estimates are functions of the red, blue, and green edges, we can bound the privacy of Algorithm~\ref{alg:smallweights} by separately considering the privacy of releasing each of the sets of colored edges. To simplify the privacy analysis, we start by considering a modified definition of adjacency, where we say that two weighted graphs $H = (V, E, W)$ and $H' = (V, E, W')$ with same unweighted graph topology are \emph{adjacent} if there is at most one differing weight $w_e \neq w_{e}'$ and $|w_e-w_{e}'| \le \Delta$, where $\Delta \le 1$ is some parameter. In this case, we will prove that \Call{BoundedWeights}{} is $(\eps \cdot \Delta, \delta)$-DP.

We will prove this via induction on the size $n = |V|$. For the base case of $n = 1$, no edges are released, so we have $(0, 0)$-DP in fact. When $n \ge 2$, assume that the algorithm chooses the parameter $T < n$. (Note that in fact we will set $T = \frac{n}{(A \eps n)^{\sqrt{2}-1}}$ or $T = \frac{n}{(A \eps n)^{(\sqrt{17}-3)/4}}$, see the commentary on Line 2 of Algorithm~\ref{alg:smallweights}.) We now proceed with the induction step.

\begin{proposition}
\label{prop:red-dp}
    Releasing the red edges is $(\eps \cdot \Delta/3K, 0)$-DP.
\end{proposition}
\begin{proof}
As the model of privacy is with respect to changing edge weights and we are assuming the sensitivity is at most $\Delta$, releasing the red edges is $(\eps \cdot \Delta/3K, 0)$-DP via the Laplace mechanism.
\end{proof}

\begin{proposition}
\label{prop:blue-approxdp}
    In the approximate-DP setting, releasing the blue edges with $Lap(10KL\sqrt{\log \delta^{-1}}/\eps)$ noise is ($\eps \cdot \Delta/3K, \delta/3K$)-DP.
\end{proposition}
\begin{proof}
The blue edges are noisy true shortest path distances between vertices in $S$. Normally, for any pair of vertices, the shortest path distance has sensitivity 1 (see the proof of Theorem~\ref{thm:unboundedpuredp}). However, because we are modifying the definition of adjacent, the shortest path distance now has sensitivity $\Delta$.
Via the Laplace mechanism, releasing a single blue edge is $(\Delta \cdot (10KL\sqrt{\log (3K/\delta)}/\eps)^{-1}, 0)$-DP.
Via strong composition, releasing all $L^2$ edges is $(\eps \cdot \Delta/3K, \delta/3K)$-DP.
\end{proof}

\begin{proposition}
\label{prop:blue-puredp}
    In the pure-DP setting, releasing the blue edges with $Lap(10KL^2/\eps)$ noise is $(\eps \cdot \Delta/10K, 0)$-DP.
\end{proposition}
\begin{proof}
    Via the Laplace mechanism, releasing a single blue edge is $O(\Delta \cdot (10KL^2/\eps)^{-1}, 0)$-DP.
    Via basic composition, releasing all $L^2$ edges is $(\eps \cdot \Delta/10K, 0)$-DP.
\end{proof}

\begin{proposition}
\label{prop:green-approxdp}
    In the approximate-DP setting, releasing the green edges is $(\eps \cdot \Delta/3K, \delta/3K)$-DP.
\end{proposition}
\begin{proof}
First,
consider a graph with no vertex $v$ s.t. $B_G(v, 100R\log n) \leq T$. In this case, there are no green edges released by the algorithm, so releasing the green edges is $(0,0)$-DP and the proposition is true.
    
Otherwise, suppose the algorithm creates balls $B^{(1, \ell)}, \dots, B^{(\tau, \ell)}$ at iteration $\ell$.
Note that the construction of the balls is independent of the edge weights, so there is no privacy loss in creating the balls.
Now, suppose $H, H'$ are adjacent graphs according to our new definition.
Then, at most one edge in all of the induced subgraphs $H[B^{(t, \ell)}]$ or $H'[B^{(t, \ell)}]$ across all $1 \le t \le \tau$ can differ, and by at most $\Delta$.
For the balls that have identical edge weights, the output distribution of the green edges are identical.
If there is some $t$ such that $H[B^{(t, \ell)}]$ or $H'[B^{(t, \ell)}]$ have different weights, we can use the induction hypothesis (since $T < n$) to say that the output of $H[B^{(t, \ell)}]$ is $(\eps \cdot \Delta/(3K^2), \delta/(3K^2))$-DP.
So overall, the outputs of $B^{(1, \ell)}, \dots, B^{(\tau, \ell)}$ together are still $(\eps \cdot \Delta/(3K^2), \delta/(3K^2))$-DP.
We repeat this procedure $K$ times since $\ell$ ranges from $1$ to $K$: by applying basic composition, we have that the green edges are $(\eps \cdot \Delta/(3K), \delta/(3K))$-DP.
\end{proof}

\begin{proposition}
\label{prop:green-puredp}
    In the pure-DP setting, releasing the green edges is $(\eps \cdot \Delta/3K, 0)$-DP.
\end{proposition}
The proof follows exactly as the proof for approximate-DP with the sole change being that
we recursively apply \Call{BoundedWeights}{} with parameters $\eps/(3K^2), 0$ instead of $\eps/(3K^2), \delta/(3K^2)$.

\begin{theorem}
    In the approximate-DP setting, Algorithm~\ref{alg:smallweights} is $(\eps \cdot \Delta, \delta)$-DP, assuming the modified definition of adjacency.
\end{theorem}
\begin{proof}
    Consider any $\tilde{D}^{(k)}$ for $k \in \{1, \ldots, K\}$. The privacy of releasing $\tilde{D}^{(k)}$ is at most the privacy of the releasing the red, blue, and green edges as $\tilde{D}^{(k)}$ is a function of those edges. By Propositions \ref{prop:red-dp}, \ref{prop:blue-approxdp}, and \ref{prop:green-approxdp}, releasing $\tilde{D}^{(k)}$ is $(\eps \cdot \Delta/K, 2\delta/3K)$-DP.
    The total privacy of the algorithm is that of releasing $\tilde{D}$ which is the median of all $k$ copies of $\tilde{D}^{(k)}$. By basic composition, the algorithm is $(\eps \cdot \Delta,\delta)$-DP, as required.
\end{proof}

\begin{theorem}
    In the pure-DP setting, Algorithm~\ref{alg:smallweights} is $(\eps \cdot \Delta, 0)$-DP, assuming the modified definition of adjacency.
\end{theorem}
\begin{proof}
    For any fixed $k$, the privacy of releasing $\tilde{D}^{(k)}$ is $(\eps \cdot \Delta/K, 0)$ private by Propositions \ref{prop:red-dp}, \ref{prop:blue-puredp}, and \ref{prop:green-puredp}.
    By basic composition over all $K$ copies, the algorithm is $(\eps \cdot \Delta,0)$-DP, as required.
\end{proof}

To finish, we return to the original definition of adjacency, where we say that two graphs $H = (V, E, W), H' = (V, E, W')$ are adjacent if $\sum_{e \in E} |w_e-w_{e'}| \le 1$. Let $m = |E|$, fix some arbitrary ordering $e_1, \dots, e_m$ of the edges in $E$, and consider the following sequence of hybrid graphs, where $H^{(0)} = H$, $H^{(m)} = H'$, and $H^{(i)}$ has the first $i$ edges with weight $w_e$ and the last $m-i$ edges with weight $w_{e'}$. For each $i \ge 1$, let $\Delta_i = |w_{e_i}-w_{e_i'}|$, meaning that $H^{(i)}, H^{(i-1)}$ are adjacent under our modified definition of adjacency if we set $\Delta = \Delta_i$.
Under the approximate-DP setting, we have that for any subset $S$ of the output, $\BP(\Call{BoundedWeights}{H^{(i)}, \eps, \delta} \in S) \le e^{\eps \cdot \Delta_i} \cdot \BP(\Call{BoundedWeights}{H^{(i-1)}, \eps, \delta} \in S) + \delta$. We can inductively apply this across all $i$ from $1$ to $m$ to get that if $H, H'$ are adjacent under the original definition,
\begin{align*}
    \BP(\Call{BoundedWeights}{H', \eps, \delta} \in S) &\le e^{\eps \cdot \sum_{i = 1}^{m} \Delta_i} \cdot \BP(\Call{BoundedWeights}{H, \eps, \delta} \in S) + m \cdot \delta \cdot e^{\eps \cdot \sum_{i = 1}^{m} \Delta_i} \\
    &\le e^{\eps} \cdot \BP(\Call{BoundedWeights}{H, \eps, \delta} \in S) + e^{\eps} \cdot n^2 \cdot \delta.
\end{align*}
Under the pure-DP setting, we have that for any subset $S$ of the output, $\BP(\Call{BoundedWeights}{H^{(i)}, \eps, 0} \in S) \le e^{\eps \cdot \Delta_i} \cdot \BP(\Call{BoundedWeights}{H^{(i-1)}, \eps, 0} \in S)$. We can inductively apply this across all $i$ from $1$ to $m$ to get that if $H, H'$ are adjacent under the original definition,
\begin{align*}
    \BP(\Call{BoundedWeights}{H', \eps, 0} \in S) &\le e^{\eps \cdot \sum_{i = 1}^{m} \Delta_i} \cdot \BP(\Call{BoundedWeights}{H, \eps, 0} \in S)\\
    &\leq e^{\eps} \cdot \BP(\Call{BoundedWeights}{H, \eps, 0} \in S).
\end{align*}

Hence, we have the following theorem, establishing privacy.

\begin{theorem}
    In the approximate-DP setting, Algorithm~\ref{alg:smallweights} is $(\eps, 3n^2 \delta)$-DP, and in the pure-DP setting, Algorithm~\ref{alg:smallweights} is $(\eps, 0)$-DP.
\end{theorem}

\begin{proof}
    This is immediate from the calculations done before and the definition of privacy, as well as the fact that we assume $\eps \le 1$ so $e^\eps \le e \le 3.$
\end{proof}

While the privacy dependence on $\delta$ is slightly worse, we can replace $\delta$ with $\delta' = \delta/(3n^2)$: as our final error dependence on $\delta$ will only be $\sqrt{\log \delta^{-1}}$ this replacement will blow up our error by a $\sqrt{\log n}$ factor at worst.

\subsection{Accuracy Analysis} \label{subsec:bounded_accuracy}

In this section, we focus on a single iteration $k$ of the outer loop (line $5$). It suffices to show that for any fixed pair $u, v \in V$, we estimate the shortest path $d(u, v)$ up to low additive error with probability at least $2/3$. In this case, by running the algorithm $K = 100 \log n$ times and outputting the median answer for each shortest path query, a simple Hoeffding bound shows that we would successfully solve all-pairs shortest path distances up to the same error with high probability.

We define $f(n, A, \eps, \delta)$ to be an upper bound for the error of each shortest path if our algorithm is $(\eps, \delta)$-DP, and if all edges are promised to be in the range $(0, A]$. We will prove a recursive bound for $f(n, A, \eps, \delta)$, and then use this to construct an explicit bound.

\begin{proposition} \label{prop:r_small}
    With probability at least $0.99$, every sampled $r_t$ is at most $100 R \log n$.
\end{proposition}

\begin{proof}
    Note that each sampled $r_t$ is more than $100 R \log n$ with probability at most $e^{-100 \log n} = n^{-100}$. Since each ball $B^{(t)}$ peels off at least one vertex, there are at most $n$ balls, so the probability that any $r_t$ is more than $100 R \log n$ is, by a union bound, at most $n \cdot n^{-100} \le 0.01$.
\end{proof}

\begin{proposition} \label{prop:red_accuracy}
    With probability at least $0.99$, every red edge $(u, v) \in E$ has the weight $w_{u, v}$ up to error $\frac{10K \log n}{\eps}.$
\end{proposition}

\begin{proof}
    Since we add Laplace noise $Lap(3K/\eps)$ to each edge, the probability that this noise exceeds $10 K \log n/\eps$ in absolute value is at most $e^{-10/3 \cdot \log n} = n^{-10/3}$. Since there are $|E| \le n^2$ total red edges in $\tilde{H}$, the probability that any red edge is off by more than $\frac{10K \log n}{\eps}$ is, by a union bound, at most $n^2 \cdot n^{-10/3} = n^{-4/3} \le 0.01$.
\end{proof}

\begin{proposition} \label{prop:blue_accuracy}
    With probability at least $0.99$, every blue edge $(u, v) \in S \times S$ has weight which equals the distance $d(u, v)$ up to error $\frac{30 K L \sqrt{\log \frac{3K}{\delta}} \log n}{\eps}$ in the approximate-DP setting, and up to error $\frac{30 KL^2 \log n}{\eps}$ in the pure-DP setting.
\end{proposition}

\begin{proof}
    We start with the approximate-DP setting. Since we add Laplace noise $Lap(10 K L \sqrt{\log \frac{3K}{\delta}}/\eps)$ to each edge, the probability that this noise exceeds $30 K L \sqrt{\log \frac{3K}{\delta}} \log n/\eps$ in absolute value is at most $e^{-3 \cdot \log n} = n^{-3}$. Since there are at most $n^2$ total blue edges in $\tilde{H}$, the probability that any blue edge is off by more than $\frac{30 K L \sqrt{\log \frac{3K}{\delta}} \log n}{\eps}$ is, by a union bound, at most $n^2 \cdot n^{-3} = n^{-1} \le 0.01$.
    
    In the pure-DP setting, we add Laplace noise $\frac{10 KL^2}{\eps}$, which is, with probability at least $1-n^{-3}$, at most $\frac{30 KL^2 \log n}{\eps}$. Since there are at most $n^2$ total blue edges in $\tilde{H}$, the probability that any blue edge is off by more than $\frac{30 K L \sqrt{\log \frac{3K}{\delta}} \log n}{\eps}$ is, by a union bound, at most $n^2 \cdot n^{-3} = n^{-1} \le 0.01$.
\end{proof}

\begin{proposition} \label{prop:green_accuracy}
    Let $f(T, A, \eps, \delta)$ be some arbitrary positive function, and assume that Algorithm \ref{alg:smallweights} is accurate up to error $f(T, A, \eps, \delta)$ for all choices of $A, \eps, \delta$ and all $T < n$. In other words, with probability at least $2/3$, if Algorithm \ref{alg:smallweights} is run on a graph with $T$ vertices, it outputs all-pairs shortest path distances up to maximum additive error $f(T, A, \eps, \delta)$.
    Then, with probability at least $0.98$, every green edge $(u, v)$, where $u, v \in B^{(t)}$ for some $t$, has weight which equals the shortest path distance from $u$ to $v$ among paths entirely contained in $B^{(t)}$, up to error $f(T, A, \frac{\eps}{3K^2}, \frac{\delta}{3K^2})$.
\end{proposition}

\begin{proof}
    We assume that every sampled $r_t$ is at most $100 R \log n$ (which by Proposition \ref{prop:r_small} occurs with at least $0.99$ probability). 
    
    In this case, every ball $B^{(t)}$ has at most $T$ vertices. So, if we run $\Call{BoundedWeights}{H[B^{(t)}], \frac{\eps}{3K^2}, \frac{\delta}{3K^2}}$, the output matrix $M^{(t, \ell)}$ will be accurate for any fixed pair $(u, v) \in B^{(t)}$ up to error $f\left(T, A, \frac{\eps}{3K^2}, \frac{\delta}{3K^2}\right)$ with probability at least $\frac{2}{3}$, since $B^{(t)}$ has at most $T$ vertices. So, by Hoeffding's inequality, by taking the entrywise median of $K = 100 \log n$ iterations, we have that $M^{(t)}$ will be accurate for any fixed pair $(u, v) \in B^{(t)}$ up to error $f\left(T, A, \frac{\eps}{3K^2}, \frac{\delta}{3K^2}\right)$, with probability at least $1-n^{-5}$. Since each pair $(u, v)$ can be in at most one ball $B^{(t)}$, we have that $M^{(t)}$ will be accurate for all $(u, v) \in B^{(t)}$ and all $t$ with probability at least $1-n^{-3} \ge 0.99$.
    
    Since we required that every sampled $r_t$ was at least $100 R \log n$, the overall probability of success is still at least $0.98$.
\end{proof}

Now, for every $u, v \in V$, we fix an arbitrary shortest path from $u$ to $v$ in the \emph{weighted} graph $H$. 


\begin{proposition} \label{prop:shortest_path_structure}
    Let the shortest path from $u$ to $v$ in $H$ be $\mathcal{P} = (p_0 = u, p_1, \dots, p_Q = v)$. Then with probability at least $0.75$, either:
\begin{enumerate}
    \item Every vertex on the path is in some removed ball $B^{(t)}$ and the number of $i < Q$ such that 
    $p_i, p_{i+1}$ are in different balls
    is at most $\frac{40 T}{R}$.
    \item There exist indices $i \le j$  with the following properties. We have $h(p_i, S) \le 100 R \log n$; for all $i' < i$, $p_{i'}$ is in some removed ball; and the number of $i' < i$ such that $p_{i'}, p_{i'+1}$ are in different balls is at most $\frac{40 T}{R}$. Also, $h(p_j, S) \le 100 R \log n$; for all $j' > j$, $p_{j'}$ is in some removed ball; and the number of $j' > j$ such that $p_{j'}, p_{j'-1}$ are in different balls is at most $\frac{40 T}{R}$.
\end{enumerate}
\end{proposition}

\begin{proof}
    First, suppose that $Q < 2T$. In this case, by Lemma \ref{lem:martingale_peeling}, the expected number of $i < Q$ such that either $p_i, p_{i+1}$ are in different balls or exactly one of $p_i, p_{i+1}$ has been peeled off is at most $\frac{8T}{R}$. So, by Markov's inequality, this number is at most $\frac{40T}{R}$ with probability at least $0.8$. Let $\mathcal{E}$ represent this event.
    
    Now, assume that event $\mathcal{E}$ holds. If every vertex on $\mathcal{P}$ is in some removed ball, we are done. Otherwise, we apply Proposition \ref{prop:peeling_hitting_set_2} to say that with failure probability at most $2n^{-100}$, there exist indices $i \le j$ such that $h(p_i, S) \le 100 R \log n$ and $h(p_j, S) \le 100 R \log n$, and that every index $i' < i$ or $j' > j$ has $p_{i'}, p_{j'}$ peeled off. In addition, if event $\mathcal{E}$ holds, then the number of $i' < i$ such that $p_{i'}, p_{i'+1}$ or or $j' > j$ such that $p_{j'}, p_{j'-1}$ are in different balls is at most $\frac{40 T}{R}.$ The overall failure probability requires either $\mathcal{E}$ to not hold, or Proposition \ref{prop:peeling_hitting_set_2} to not hold, which is at most $0.2 + 2n^{-100} \le 0.25$. This concludes the proof when $Q < 2T$.
    
    Next, suppose that $Q \ge 2T$. In this case, we consider $\mathcal{P}_{\text{left}}$ to be the path just consisting of $p_0, \dots, p_{T}$, and $\mathcal{P}_{\text{right}}$ to be the path just consisting of $p_{Q-T}, \dots, p_{Q}$. We again use Lemma \ref{lem:martingale_peeling} to say the expected number of $i < T$ such that $p_i, p_{i+1}$ are in different balls or exactly one of $p_i, p_{i+1}$ has been peeled off is at most $\frac{4T}{R}$. Similarly, we use Lemma \ref{lem:martingale_peeling} to say the expected number of $i > Q-T$ such that $p_i, p_{i-1}$ are in different balls or exactly one of $p_i, p_{i-1}$ has been peeled off is at most $\frac{4T}{R}$. So, by Markov's inequality, each of these numbers is at most $\frac{40 T}{R}$ with probability at least $0.9$, so both are at most $\frac{40 T}{R}$ with probability at least $0.8$. Let $\mathcal{E}'$ represent this event.
    
    We can apply Proposition \ref{prop:peeling_hitting_set} to say that with failure probability at most $n^{-100}$, there is an index $i \le T$ such that $h(p_i, S) \le 100 R \log n$ and every vertex $p_{i'}$ for $i' < i$ has been peeled off. Likewise, with failure probability $n^{-100}$, there is an index $j \ge Q-T$ such that $h(p_j, S) \le 100 R \log n$ and every vertex $p_{j'}$ for $j' > j$ has been peeled off. Since $Q \ge 2T$, we have that $j \ge i$. So, with failure probability at most $0.2 + 2n^{-100} \le 0.25$, both $\mathcal{E}'$ and the two events described in the previous sentences hold. This concludes the proof when $Q \ge 2T$.
\end{proof}

\begin{lemma} \label{lem:accuracy_when_no_noise}
    Let $H = (V, E, W)$ and $u, v \in V$ be fixed, and suppose that the event in Proposition \ref{prop:shortest_path_structure} holds. Consider a variant of the main algorithm where we create a multigraph $\hat{H}$ (instead of $\tilde{H}$) such that we do not add Laplace noise to the red or blue edges, and do not make the green edges private.
    (In other words, the edge weights are entirely accurate). Then, the absolute difference between the length of the true shortest path in $H$, $d(u, v)$, and the length of the shortest path in the multigraph from $u$ to $v$ using at most $200 R \log n + \frac{100 T}{R}$ red edges, at most $1$ blue edge, and at most $\frac{100 T}{R}$ green edges, is at most $400 R \log n \cdot A$.
\end{lemma}

\begin{proof}
    Since all edges are now fully accurate, every edge weight in $\hat{H}$ from some vertex $v_1$ to some $v_2$ is at least  $d(v_1, v_2)$. So, the shortest path in $\hat{H}$ from $u$ to $v$ is at least $d(u, v)$. Let $\mathcal{P} = (p_0 = u, p_1, \dots, p_Q = v)$ be the shortest path from $u$ to $v$ in $H$.
    
    First, suppose that every vertex on $\mathcal{P}$ is in some removed ball and the number of $i$ such that $p_i, p_{i+1}$ are in different balls is at most $\frac{40 T}{R}.$ So, the set $\{0, 1, \dots, Q\}$ can be partitioned into $g \le \frac{40 T}{R}$ intervals $\{a_1 = 0, \dots, b_1\}, \{a_2 = b_1+1, \dots, b_2\}, \dots, \{a_g, \dots, b_g = Q\},$ where the vertices corresponding to each interval are all in the same ball $B^{(t)}$. This means that for each $1 \le i \le g$, there is a green edge from $p_{a_i}$ to $p_{b_i}$ with weight at most the sum of the weights on the path $\mathcal{P}$ from $p_{a_i}$ to $p_{b_i}$. This is true because this piece of the path is entirely contained in the induced subgraph $H[B^{(t)}]$. So, by connecting $p_{a_i}$ to $p_{b_i}$ with a green edge and connecting $p_{b_i}$ to $p_{a_{i+1}} = p_{b_i+1}$ with a red edge, we can connect $p_0 = u$ to $p_Q = v$ using at most $\frac{40 T}{R}$ green and $\frac{40 T}{R}$ red edges. In addition, this path is at most the length of the shortest path in $H$ from $u$ to $v$ (and in fact must be equal). See the top of Figure \ref{fig:main_fig} for a diagram representation of this case.
    
    Alternatively, because the event in Proposition \ref{prop:shortest_path_structure} holds, there exist $i \le j$ such that $p_{i'}$ is in some removed ball for all $i' < i$, and the number of $i' < i$ such that $p_{i'}, p_{i'+1}$ are in different balls is at most $\frac{40 T}{R}.$ Likewise, $p_{j'}$ is in some removed ball for all $j' > j$, and the number of $j' > j$ such that $p_{j'}, p_{j'-1}$ are in different balls is at most $\frac{40 T}{R}.$
    So, by the same argument as the paragraph above, there exists a path in $\hat{H}$ from $u$ to $p_{i}$ using at most $\frac{40 T}{R}$ green and $\frac{40 T}{R}$ red edges, which has total length at most $d(u, p_i)$, and a path in $\hat{H}$ from $p_j$ to $v$ using at most $\frac{40 T}{R}$ green and $\frac{40 T}{R}$ red edges, which has total length at most $d(p_j, v)$.
    In addition, there is some $s_1 \in S$ such that $h(p_i, s_1) \le 100 R \log n$, so we can connect $p_i$ and $s_1$ using at most $100 R \log n$ red edges. Likewise, there is some $s_2 \in S$ such that $p_j$ and $s_2$ can be connected using at most $100 R \log n$ edges.
    
    So, we can consider the path using at most $\frac{40 T}{R}$ red and $\frac{40 T}{R}$ green edges from $u$ to $p_i$, then using at most $100 R \log n$ red edges from $p_i$ to $s_1$, then one blue edge from $s_1$ to $s_2$ with weight exactly $d(s_1, s_2)$ (since $s_1, s_2 \in S$), then at most $100 R \log n$ red edges from $s_2$ to $p_j$, and finally at most $\frac{40 T}{R}$ red and $\frac{40 T}{R}$ green edges from $p_j$ to $v$. See the bottom of Figure \ref{fig:main_fig} for a diagram representation of this case.
    
    Overall, this path has at most $200 R \log n + \frac{100 T}{R}$ red edges, at most $1$ blue edge, and at most $\frac{100 T}{R}$ green edges. In addition, the total length of this path in $H$ is at most
\begin{align*}
    &\hspace{0.5cm} d(u, p_i) + 100 R \log n \cdot A + d(s_1, s_2) + 100 R \log n \cdot A + d(p_j, v) \\
    &= d(u, p_i) + d(s_1, s_2) + d(p_j, v) + 200 R \log n \cdot A \\
    &\le d(u, p_i) + d(s_1, p_i) + d(p_i, p_j) + d(p_j, s_2) + d(p_j, v) + 200 R \log n \cdot A \\
    &\le d(u, p_i) + d(p_i, p_j) + d(p_j, v) + 200 R \log n \cdot A + 200 R \log n \cdot A \\
    &= d(u, v) + 400 R \log n \cdot A.
\end{align*}
    Above, we use the fact that every red edge has weight at most $A$, which also means the distance from $p_i$ to $s_1$ and from $p_j$ to $s_2$ are at most $100 R \log n \cdot A$, and the fact  that the distance from $s_1$ to $s_2$ is upper bounded by the length of the three shortest paths that go from $s_1$ to $p_i$ to $p_j$ to $s_2$.
\end{proof}

\begin{figure}[tp]
    \centering
\begin{tikzpicture}[scale=1.5]
    \node (u) at (0,0) {$u$};
    \node (v) at (8.7,0) {$v$};
    
    \draw[black,dashed] (u) -- (v);
    
    \draw[green, thick, bend left] (0,0) edge (1.2, 0);
    \draw[red, thick, bend left] (1.2,0) edge (1.5, 0);
    \draw[green, thick, bend left] (1.5,0) edge (2.4, 0);
    \draw[red, thick, bend left] (2.4,0) edge (2.7, 0);
    \draw[green, thick, bend left] (2.7,0) edge (4.2, 0);
    \draw[red, thick, bend left] (4.2,0) edge (4.5, 0);
    \draw[green, thick, bend left] (4.5,0) edge (6.2, 0);
    \draw[red, thick, bend left] (6.2,0) edge (6.5, 0);
    \draw[green, thick, bend left] (6.5,0) edge (7.3, 0);
    \draw[red, thick, bend left] (7.3,0) edge (7.6, 0);
    \draw[green, thick, bend left] (7.6,0) edge (8.7, 0);
\end{tikzpicture}

\hspace{0.5cm}

\begin{tikzpicture}[scale=1.5]
    \node (u) at (0,0) {$u$};
    \node (v) at (8.7,0) {$v$};
    \node (pi) at (2.7,0) {$p_i$};
    \node (pj) at (7.3,0) {$p_j$};
    \node (s1) at (2.7, 3) {$s_1$};
    \node (s2) at (7.3, 3) {$s_2$};
    
    \draw[black,dashed] (u) -- (pi);
    \draw[black,dashed] (pi) -- (pj);
    \draw[black,dashed] (pj) -- (v);
    
    \draw[green, thick, bend left] (0,0) edge (1.2, 0);
    \draw[red, thick, bend left] (1.2,0) edge (1.5, 0);
    \draw[green, thick, bend left] (1.5,0) edge (2.4, 0);
    \draw[red, thick, bend left] (2.4,0) edge (2.6,0.04);

    \draw[red, thick, bend left] (7.4,0.04) edge (7.6, 0);
    \draw[green, thick, bend left] (7.6,0) edge (8.7, 0);
    
    \draw[red, thick, dashed] (pi) edge (s1);
    \draw[red, thick, dashed] (pj) edge (s2);
    \draw[blue, thick] (s1) edge (s2);
\end{tikzpicture}
    \caption{Representation of the two cases of Lemma \ref{lem:accuracy_when_no_noise}. The top represents when every point on the shortest path from $u$ to $v$ in $H$ has been peeled and the number of times the ball $B^{(t)}$ changes along this path is at most $\frac{40 T}{R}.$ In this case, we can directly use green and red edges. The bottom represents the remaining cases, in which we can find some $p_i$ that is ``close'' to $u$ on this path, and some $p_j$ that is ``close'' to $v$ on this path, such that each of $p_i, p_j$ are ``close'' to $S$. In this case, we alternate green and red edges between $u$ and $p_i$ and between $p_j$ and $v$, use only red edges from $p_i$ to $s_1$ and from $s_2$ to $p_j$, and use a single blue edge from $s_1$ to $s_2$.}
    \label{fig:main_fig}
\end{figure}
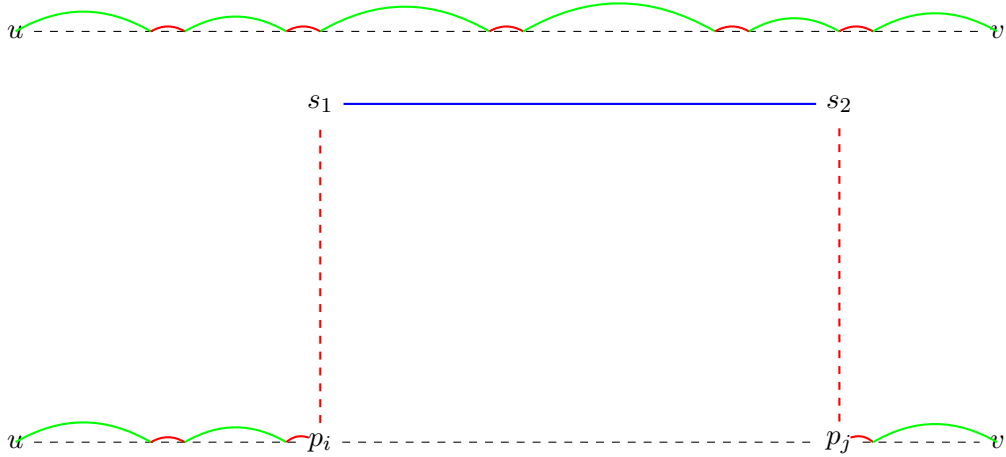

We now focus on the $(\eps, \delta)$-DP case. At the end of this section, we show how to modify the results to also get an algorithm for the pure $(\eps, 0)$-DP case.

\begin{theorem} \label{thm:result_up_to_recursion}
    Let $H = (V, E, W)$ and $u, v \in V$ be fixed. Suppose that the events in Propositions \ref{prop:red_accuracy}, \ref{prop:blue_accuracy}, \ref{prop:green_accuracy}, and \ref{prop:shortest_path_structure} all hold. Also, suppose that $f(n, A, \eps, \delta)$ has at most linear dependence on $\eps^{-1}$ and logarithmic dependence on $\delta^{-1}$, meaning that $f(n, A, \eps/N_1, \delta/N_2) \le N_1 \cdot \log N_2 \cdot f(n, A, \eps, \delta)$ for any $N_1 \ge 1$, $N_2 \ge 10$. Then, our estimate $\tilde{D}^{(k)}[u, v]$ is accurate up to additive error
\[
(\log n)^{O(1)} \cdot \max\left(\frac{\sqrt{\log \delta^{-1}}}{\eps} \cdot \max\left(R, \frac{T}{R}, \frac{n}{T}\right), \frac{T}{R} \cdot f(T, A, \eps, \delta), A \cdot R\right).
\]
\end{theorem}

\begin{proof}
    By Lemma \ref{lem:accuracy_when_no_noise}, we have that if we did not add Laplace errors to the red edges or blue edges, and returned exact distances in the balls for the green edges, our output would be accurate up to error $400 R \log n \cdot A$. Now, by Propositions \ref{prop:red_accuracy}, \ref{prop:blue_accuracy}, and \ref{prop:green_accuracy}, we have that \emph{every} red edge is accurate up to additive error $\frac{10 K \log n}{\eps}$, \emph{every} blue edge is accurate up to error $\frac{30 K L \sqrt{\log (3K/\delta)} \log n}{\eps},$ and \emph{every} green edge is accurate up to error $f(T, A, \frac{\eps}{3K^2}, \frac{\delta}{3K^2})$. Hence, every path that uses at most $100 R \log n + \frac{100 T}{R}$ red edges, at most $1$ blue edge, and $\frac{100 T}{R}$ green edges does not change in total weight by more than
\begin{align*}
    &\hspace{0.5cm} \frac{10 K \log n}{\eps} \cdot \left(100 R \log n + \frac{100 T}{R}\right) + \frac{30 K L \sqrt{\log \frac{3K}{\delta}} \log n}{\eps} \cdot 1 + f(T, A, \frac{\eps}{3K^2}, \frac{\delta}{3K^2}) \cdot \frac{100 T}{R} \\
    &\le (\log n)^{O(1)} \cdot \frac{\sqrt{\log \delta^{-1}}}{\eps} \left(R + \frac{T}{R} + \frac{n}{T}\right) + f(T, A, \frac{\eps}{3K^2}, \frac{\delta}{3K^2}) \cdot \frac{100 T}{R} \\
    &\le (\log n)^{O(1)} \cdot \left[\frac{\sqrt{\log \delta^{-1}}}{\eps} \left(R + \frac{T}{R} + \frac{n}{T}\right) + f(T, A, \frac{\eps}{3K^2}, \frac{\delta}{3K^2}) \cdot \frac{T}{R}\right],
\end{align*}
    where we recall that $L = 100 \log n \cdot \frac{n}{T}$ and $K = 100\log n$. Adding this to the $400 R \log n \cdot A$ error from Lemma \ref{lem:accuracy_when_no_noise}, we have that $\tilde{D}^{(k)}[u, v]$ is accurate up to error
\begin{align*}
    &\hspace{0.5cm} (\log n)^{O(1)} \cdot \left[\frac{\sqrt{\log \delta^{-1}}}{\eps} \left(R + \frac{T}{R} + \frac{n}{T}\right) + f(T, A, \frac{\eps}{3K^2}, \frac{\delta}{3K^2}) \cdot \frac{T}{R} + A \cdot R\right] \\
    &\le (\log n)^{O(1)} \cdot \left[\frac{\sqrt{\log \delta^{-1}}}{\eps} \left(R + \frac{T}{R} + \frac{n}{T}\right) + f(T, A, \eps, \delta) \cdot \frac{T}{R} + A \cdot R\right] \\
    &\le (\log n)^{O(1)} \cdot \max\left(\frac{\sqrt{\log \delta^{-1}}}{\eps} \max\left(R, \frac{T}{R}, \frac{n}{T}\right), f(T, A, \eps, \delta) \cdot \frac{T}{R}, A \cdot R\right).
\end{align*}
    The penultimate inequality follows since $K = O(\log n)$ and since $f$ has at most linear dependence on $\eps^{-1}$ and logarithmic dependence on $\delta^{-1}$. The final inequality follows since the sum and max of $O(1)$ positive terms are asymptotically equivalent.
\end{proof}

Now, define $\alpha = \sqrt{2}-1$. Note that $\alpha$ is a root of $1-2\alpha-\alpha^2 = 0$. Suppose that the parameters $A, \eps, \delta$ are fixed, and that $A \ge \eps^{-1}$.  Now, let us consider the recursive function
\[f_0(n) = \max\left(\frac{\sqrt{\log \delta^{-1}}}{\eps} \cdot \max\left(R, \frac{T}{R}, \frac{n}{T}\right), \frac{T}{R} \cdot f_0(T), A \cdot R\right),\]
where $R = \frac{A^{\alpha-1}n^\alpha}{\eps^{1-\alpha}} = \frac{n}{(A \eps n)^{1-\alpha}}$ and $T = \frac{A^{-\alpha}n^{1-\alpha}}{\eps^{\alpha}} = \frac{n}{(A \eps n)^\alpha}$.
(Note that this is indeed recursive since $T = \frac{n}{(A \eps n)^\alpha} \le \frac{n}{n^\alpha} \le n^{2-\sqrt{2}}$.)
We will show that $f_0(n) \le \sqrt{\log \delta^{-1}} \cdot \frac{A^{\alpha} n^{\alpha}}{\eps^{1-\alpha}}.$

First, note that 
\begin{equation} \label{eq:final_1}
    \frac{\sqrt{\log \delta^{-1}}}{\eps} \cdot R = \sqrt{\log \delta^{-1}} \cdot \eps^{-1} \cdot \frac{A^{\alpha-1} n^\alpha}{\eps^{1-\alpha}} \le \sqrt{\log \delta^{-1}} \cdot A \cdot \frac{A^{\alpha-1} n^\alpha}{\eps^{1-\alpha}} = \sqrt{\log \delta^{-1}} \cdot \frac{A^\alpha n^\alpha}{\eps^{1-\alpha}}.
\end{equation}
Next, note that 
\begin{equation} \label{eq:final_2}
    \frac{\sqrt{\log \delta^{-1}}}{\eps} \cdot \frac{T}{R} = \sqrt{\log \delta^{-1}} \cdot \eps^{-1} \cdot (A \eps n)^{1-2 \alpha} \le \sqrt{\log \delta^{-1}} \cdot \eps^{-1} \cdot (A \eps n)^{\alpha} = \sqrt{\log \delta^{-1}} \cdot \frac{A^\alpha n^\alpha}{\eps^{1-\alpha}}.
\end{equation}
Next, note that
\begin{equation} \label{eq:final_3}
    \frac{\sqrt{\log \delta^{-1}}}{\eps} \cdot \frac{n}{T} = \frac{\sqrt{\log \delta^{-1}}}{\eps} \cdot (A \eps n)^\alpha = \sqrt{\log \delta^{-1}} \cdot \frac{A^\alpha n^\alpha}{\eps^{1-\alpha}}.
\end{equation}
Next, note that by inductively applying $f_0(T) \le \sqrt{\log \delta^{-1}} \cdot \frac{A^\alpha T^\alpha}{\eps^{1-\alpha}},$ we have
\begin{multline} \label{eq:final_4}
    \frac{T}{R} \cdot f_0(T) \le (A \eps n)^{1-2 \alpha} \cdot \sqrt{\log \delta^{-1}} \cdot \frac{A^\alpha T^\alpha}{\eps^{1-\alpha}} = \sqrt{\log \delta^{-1}} \cdot (A \eps n)^{1-2 \alpha} \cdot \frac{A^\alpha}{\eps^{1-\alpha}} \cdot \left(\frac{n}{(A \eps n)^\alpha}\right)^\alpha \\
    = \sqrt{\log \delta^{-1}} \cdot (A \eps n)^{1-2 \alpha - \alpha^2} \cdot \frac{A^\alpha}{\eps^{1-\alpha}} \cdot n^\alpha = \sqrt{\log \delta^{-1}} \cdot \frac{A^\alpha n^\alpha}{\eps^{1-\alpha}}.
\end{multline}
    Finally, we have that
\begin{equation} \label{eq:final_5}
    A \cdot R = A \cdot \frac{A^{\alpha-1}n^\alpha}{\eps^{1-\alpha}} = \frac{A^\alpha n^\alpha}{\eps^{1-\alpha}}.
\end{equation}
    So, combining Equations \eqref{eq:final_1}, \eqref{eq:final_2}, \eqref{eq:final_3}, \eqref{eq:final_4}, and \eqref{eq:final_5}, we have that $f_0(n) \le \sqrt{\log \delta^{-1}} \cdot \frac{A^{\alpha} n^{\alpha}}{\eps^{1-\alpha}}.$
    
    In reality, we wish for the recursion
\begin{equation} \label{eq:recursion_approx_dp}
f(n) \le (\log n)^C \cdot \max\left(\frac{\sqrt{\log \delta^{-1}}}{\eps} \cdot \max\left(R, \frac{T}{R}, \frac{n}{T}\right), \frac{T}{R} \cdot f(T), A \cdot R\right)
\end{equation}
    for some constant $C$.
    However, note if a function $g(n)$ satisfies the recursion $g(n) \le g(T) \cdot (\log n)^C$ for $T = \frac{n}{(A \eps n)^\alpha} \le n^{2-\sqrt{2}}$, then $g(n) \le (\log n)^{10 C \log \log n}$, since the number of times we call the recursion is at most $10 \log \log n$ until $T$ becomes at most $2$. So, by letting $f(n) = f_0(n) \cdot g(n)$, we have that $f(n)$ satisfies the desired recursion. Hence, we have that the error of our algorithm is at most
\[f_0(n) \cdot (\log n)^{10 C \log \log n} \le \sqrt{\log \delta^{-1}} \cdot \frac{A^\alpha n^{\alpha+o(1)}}{\eps^{1-\alpha}},\]
    where $\alpha = \sqrt{2}-1$. This was all assuming that $A \ge \eps^{-1}$, so if $A < \eps^{-1}$, we can instead replace $A$ with $\eps^{-1}$ to get the bound $\sqrt{\log \delta^{-1}} \cdot \frac{\eps^{-\alpha} n^{\alpha+o(1)}}{\eps^{1-\alpha}} = \sqrt{\log \delta^{-1}} \cdot \eps^{-1} n^{\alpha+o(1)}.$ Hence, we have the following theorem.
    
\begin{theorem}[Approximate-DP Setting] \label{thm:main_bounded_approx}
    With probability at least $\frac{2}{3}$, the estimate Algorithm \ref{alg:smallweights} estimates the true shortest path $d(u, v)$ for every $u, v \in V$ up to additive error
\begin{equation} \label{eq:final_approx}
    \sqrt{\log \delta^{-1}} \cdot \frac{A^\alpha n^{\alpha+o(1)}}{\eps^{1-\alpha}} + \sqrt{\log \delta^{-1}} \cdot \eps^{-1} n^{\alpha+o(1)},
\end{equation}
    where $\alpha = \sqrt{2}-1 \le 0.4143$.
\end{theorem}

\begin{proof}
    Let $f(n, A, \eps, \delta)$ be the expression in \eqref{eq:final_approx} for some appropriate choice of $n^{o(1)}$ 
    such that it satisfies the recursion in \eqref{eq:recursion_approx_dp}. We proceed by strong induction, and prove that with probability at least $2/3$, the maximum error is at most $f(n, A, \eps, \delta)$.
    
    We start by focusing on a single pair $(u, v) \in V$ and a single iteration of the outer loop, which outputs some estimate $\mathcal{D}^{(k)}$.
    Note that $f$ has at most linear dependence on $\eps^{-1}$ and at most logarithmic dependence on $\delta^{-1}$, so we can apply Theorem \ref{thm:result_up_to_recursion}. This requires Propositions \ref{prop:red_accuracy}, \ref{prop:blue_accuracy}, \ref{prop:green_accuracy} (which can be applied by our induction hypothesis as $T \le n^{2-\sqrt{2}}$), and \ref{prop:shortest_path_structure} to all hold, which by a union bound holds with probability at least $0.71 \ge \frac{2}{3}$. Since $f(n, A, \eps, \delta)$ satisfies the recursion of Equation \eqref{eq:recursion_approx_dp}, by Theorem \ref{thm:result_up_to_recursion}, the probability that $\tilde{D}^{(k)}[u, v]$ accurately estimates the true distance $d(u, v)$ up to error $f(n, A, \eps, \delta)$ is at least $2/3$.
    
    Hence, because our final estimate takes the entrywise median of $K = 100 \log n$ copies of $\tilde{D}^{(k)}$, a simple Hoeffding bound implies that for all $u, v \in V$, $\BP\left(|\tilde{D}[u, v]-d(u, v)| \ge f(n, A, \eps, \delta)\right) \le \frac{1}{n^3}$, so this holds simultaneously for all $u, v \in V$ with probability at least $1-\frac{1}{n} \ge \frac{2}{3}$.
\end{proof}

We now focus on the pure-DP setting, where we wish for $(\eps, 0)$-DP. In this case, Theorem \ref{thm:result_up_to_recursion} still holds, except that the blue edges are now accurate up to error $\frac{30 K L^2 \log n}{\eps}$ instead of $\frac{30 K L \sqrt{\log \delta^{-1}} \log n}{\eps}$. Using this, we obtain that now the estimate $\tilde{D}^{(k)}[u, v]$ is accurate up to error
\[(\log n)^C \cdot \max\left(\frac{1}{\eps} \cdot \max\left(R, \frac{T}{R}, \frac{n^2}{T^2}\right), \frac{T}{R} \cdot f(T, A, \eps), A \cdot R\right),\]
assuming that on graphs with $T$ nodes, the algorithm is accurate up to error $f(T, A, \eps)$.

Again, we assume first that $\eps^{-1} \le A$ and consider the simpler recursion
\[f_0(n) \le \max\left(\frac{1}{\eps} \cdot \max\left(R, \frac{T}{R}, \frac{n^2}{T^2}\right), \frac{T}{R} \cdot f_0(T), A \cdot R\right).\]
Indeed, this recursion is satisfied by the function $f_0(n) = (A n)^\alpha \cdot \eps^{-(1-\alpha)}$ for $\alpha = \frac{\sqrt{17}-3}{2}$, when $T = n^{1-\alpha/2} A^{-\alpha/2} \eps^{-\alpha/2} = \frac{n}{(A \eps n)^{\alpha/2}}$ and $R = n^\alpha A^{\alpha-1} \eps^{\alpha-1} = \frac{n}{(A \eps n)^{1-\alpha}}$. Note that $\alpha$ is a root of $1-\frac{3 \alpha}{2} - \frac{\alpha^2}{2} = 0.$ Again, we have that $A \eps \ge 1$, so $T \le n^{1-\alpha / 2}$; hence, this function is in fact recursive.

To verify the recursion, first note that $\frac{R}{\eps} \le A \cdot R$ so we can ignore the $\frac{1}{\eps} \cdot R$ term. For the rest of the terms, note that
\begin{equation} \label{eq:final_pure_1}
    \frac{1}{\eps} \cdot \frac{T}{R} = \frac{1}{\eps} \cdot (A \eps n)^{1 - 3\alpha/2} \le \frac{1}{\eps} \cdot (A \eps n)^{\alpha} = (An)^\alpha \eps^{-(1-\alpha)},
\end{equation}
\begin{equation} \label{eq:final_pure_2}
    \frac{1}{\eps} \cdot \frac{n^2}{T^2} = \frac{1}{\eps} \cdot (A \eps n)^{\alpha} = (An)^\alpha \eps^{-(1-\alpha)},
\end{equation}
\begin{multline} \label{eq:final_pure_3}
    \frac{T}{R} \cdot f_0(T) = (A \eps n)^{1 - 3\alpha/2} \cdot (AT)^\alpha \eps^{-(1-\alpha)} = (A \eps n)^{1 - 3\alpha/2} \cdot A^\alpha \eps^{-(1-\alpha)} \cdot \left(\frac{n}{(A \eps n)^{\alpha/2}}\right)^\alpha \\
    = (A \eps n)^{1-3 \alpha/2 - \alpha^2/2} \cdot A^\alpha \eps^{-(1-\alpha)} n^\alpha = (An)^\alpha \eps^{-(1-\alpha)},
\end{multline}
and
\begin{equation} \label{eq:final_pure_4}
    A \cdot R = A \cdot n^\alpha A^{\alpha-1} \eps^{\alpha-1} = (An)^\alpha \eps^{-(1-\alpha)}.
\end{equation}
So, combining Equations \eqref{eq:final_pure_1}, \eqref{eq:final_pure_2}, \eqref{eq:final_pure_3}, and \eqref{eq:final_pure_4}, we have that $f_0(n) \le (An)^\alpha \eps^{-(\alpha-1)} = \frac{A^\alpha n^\alpha}{\eps^{1-\alpha}}$.

Hence, we may replace $f_0(n)$ with $f(n)$ and apply the same argument as in Theorem \ref{thm:main_bounded_approx} to obtain the following theorem.

\begin{theorem}[Pure-DP Setting] \label{thm:main_bounded_pure}
    With probability at least $\frac{2}{3}$, the estimate Algorithm \ref{alg:smallweights} estimates the true shortest path $d(u, v)$ for every $u, v \in V$ up to additive error
\begin{equation} \label{eq:final_pure}
    \frac{A^\alpha n^{\alpha+o(1)}}{\eps^{1-\alpha}} + \eps^{-1} n^{\alpha+o(1)},
\end{equation}
    where $\alpha = \frac{\sqrt{17}-3}{2} < 0.5616$.
\end{theorem}

To summarize, when $A$ and $\eps^{-1}, \log \delta^{-1}$ are sufficiently low (i.e., at most $n^{o(1)}$), we can solve all-pairs shortest path distances up to error $n^{\sqrt{2}-1 + o(1)} < O(n^{0.4143})$ with $(\eps, \delta)$-DP, and up to error $n^{(\sqrt{17}-3)/2 + o(1)} < O(n^{0.5616})$ with $(\eps, 0)$-DP.

\section*{Acknowledgments} \label{sec:ack}
The authors would like to thank Anders Aamand, Piotr Indyk, and Sandeep Silwal for helpful discussions regarding this problem.

\newcommand{\etalchar}[1]{$^{#1}$}


\begin{thebibliography}{GKMN22}

\bibitem[ACIM99]{AingworthCIM99}
Donald Aingworth, Chandra Chekuri, Piotr Indyk, and Rajeev Motwani.
\newblock Fast estimation of diameter and shortest paths (without matrix
  multiplication).
\newblock {\em {SIAM} J. Comput.}, 28(4):1167--1181, 1999.

\bibitem[AV21]{AlmanV}
Josh Alman and Virginia {Vassilevska Williams}.
\newblock A refined laser method and faster matrix multiplication.
\newblock In {\em Proceedings of the 2021 ACM-SIAM Symposium on Discrete
  Algorithms (SODA)}, pages 522--539, 2021.

\bibitem[BCGO16]{brunet2016privategraph}
Solenn Brunet, S{\'{e}}bastien Canard, S{\'{e}}bastien Gambs, and Baptiste
  Olivier.
\newblock Edge-calibrated noise for differentially private mechanisms on
  graphs.
\newblock In {\em 14th Annual Conference on Privacy, Security and Trust (PST)},
  pages 42--49, 2016.

\bibitem[BKMP10]{BaswanaKMP10}
Surender Baswana, Telikepalli Kavitha, Kurt Mehlhorn, and Seth Pettie.
\newblock Additive spanners and ($\alpha, \beta$)-spanners.
\newblock {\em ACM Trans. Algorithms}, 7(1), Dec 2010.

\bibitem[Cha10]{Chan10}
Timothy~M. Chan.
\newblock More algorithms for all-pairs shortest paths in weighted graphs.
\newblock {\em {SIAM} J. Comput.}, 39(5):2075--2089, 2010.

\bibitem[Che13]{Chechik13}
Shiri Chechik.
\newblock New additive spanners.
\newblock In {\em Proceedings of the 2013 Annual ACM-SIAM Symposium on Discrete
  Algorithms (SODA)}, pages 498--512, 2013.

\bibitem[DHZ00]{DorHZ00}
Dorit Dor, Shay Halperin, and Uri Zwick.
\newblock All-pairs almost shortest paths.
\newblock {\em {SIAM} J. Comput.}, 29(5):1740--1759, 2000.

\bibitem[DKY17]{Microsoft}
Bolin Ding, Janardhan Kulkarni, and Sergey Yekhanin.
\newblock Collecting telemetry data privately.
\newblock In {\em Advances in Neural Information Processing Systems}, pages
  3571--3580, 2017.

\bibitem[DLS{\etalchar{+}}17]{Census}
Aref~N. Dajani, Amy~D. Lauger, Phyllis~E. Singer, Daniel Kifer, Jerome~P.
  Reiter, Ashwin Machanavajjhala, Simson~L. Garfinkel, Scot~A. Dahl, Matthew
  Graham, Vishesh Karwa, Hang Kim, Philip Lelerc, Ian~M. Schmutte, William~N.
  Sexton, Lars Vilhuber, and John~M. Abowd.
\newblock The modernization of statistical disclosure limitation at the u.s.
  {C}ensus {B}ureau, 2017.
\newblock In {\em Presented at the September 2017 meeting of the Census
  Scientific Advisory Committee}, 2017.

\bibitem[DMNS06]{dwork2006dp}
Cynthia Dwork, Frank McSherry, Kobbi Nissim, and Adam~D. Smith.
\newblock Calibrating noise to sensitivity in private data analysis.
\newblock In {\em Theory of Cryptography Conference (TCC)}, volume 3876 of {\em
  Lecture Notes in Computer Science}, pages 265--284, 2006.

\bibitem[DR14]{dworkrothbook}
Cynthia Dwork and Aaron Roth.
\newblock The algorithmic foundations of differential privacy.
\newblock {\em Found. Trends Theor. Comput. Sci.}, 9(3-4):211--407, 2014.

\bibitem[EPK14]{Google}
{\'{U}}lfar Erlingsson, Vasyl Pihur, and Aleksandra Korolova.
\newblock {RAPPOR:} randomized aggregatable privacy-preserving ordinal
  response.
\newblock In {\em Proceedings of the 2014 {ACM} {SIGSAC} Conference on Computer
  and Communications Security (CCS)}, pages 1054--1067, 2014.

\bibitem[Fre76]{fredman1976new}
Michael~L. Fredman.
\newblock New bounds on the complexity of the shortest path problem.
\newblock {\em SIAM J. Comput.}, 5(1):83--89, 1976.

\bibitem[GKMN22]{GhaziScooped}
Badih Ghazi, Ravi Kumar, Pasin Manurangsi, and Jelani Nelson.
\newblock Differentially private all-pairs shortest path distances: Improved
  algorithms and lower bounds.
\newblock {\em CoRR}, abs/2203.16476, 2022.

\bibitem[GM97a]{GalilM97distance}
Zvi Galil and Oded Margalit.
\newblock All pairs shortest distances for graphs with small integer length
  edges.
\newblock {\em Inf. Comput.}, 134(2):103--139, 1997.

\bibitem[GM97b]{GalilM97path}
Zvi Galil and Oded Margalit.
\newblock All pairs shortest paths for graphs with small integer length edges.
\newblock {\em J. Comput. Syst. Sci.}, 54(2):243--254, 1997.

\bibitem[GS01]{probability_book}
Geoffrey~R. Grimmett and David~R. Stirzaker.
\newblock {\em Probability and Random Processes}.
\newblock Oxford University Press, 2001.

\bibitem[HLMJ09]{hay2009graphdp}
Michael Hay, Chao Li, Gerome Miklau, and David~D. Jensen.
\newblock Accurate estimation of the degree distribution of private networks.
\newblock In {\em Ninth {IEEE} International Conference on Data Mining (ICDM)},
  pages 169--178, 2009.

\bibitem[HT12]{HanTakaoka}
Yijie Han and Tadao Takaoka.
\newblock An ${O}(n^3 \log\log n/\log^2 n)$ time algorithm for all pairs
  shortest paths.
\newblock In {\em 13th Scandinavian Symposium and Workshops on Algorithm Theory
  (SWAT)}, pages 131--141, 2012.

\bibitem[LU18]{legallurr}
Fran\c{c}ois {Le Gall} and Florent Urrutia.
\newblock Improved rectangular matrix multiplication using powers of the
  coppersmith-winograd tensor.
\newblock In {\em Proceedings of the Twenty-Ninth Annual ACM-SIAM Symposium on
  Discrete Algorithms (SODA)}, page 1029–1046, 2018.

\bibitem[Pin18]{pinot2018mst}
Rafael Pinot.
\newblock Minimum spanning tree release under differential privacy constraints.
\newblock {\em CoRR}, abs/1801.06423, 2018.

\bibitem[PMY{\etalchar{+}}18]{pinot2018graphclustering}
Rafael Pinot, Anne Morvan, Florian Yger, C{\'{e}}dric Gouy{-}Pailler, and Jamal
  Atif.
\newblock Graph-based clustering under differential privacy.
\newblock In {\em Proceedings of the Thirty-Fourth Conference on Uncertainty in
  Artificial Intelligence (UAI)}, pages 329--338, 2018.

\bibitem[Sea16]{sealfon2016dpapsp}
Adam Sealfon.
\newblock Shortest paths and distances with differential privacy.
\newblock In {\em {Proceedings of the 35th ACM SIGMOD-SIGACT-SIGAI Symposium on
  Principles of Database Systems (PODS)}}. Association for Computing Machinery,
  2016.

\bibitem[Sea20]{DPorg-open-problem-all-pairs}
Adam Sealfon.
\newblock Open problem - private all-pairs distances.
\newblock DifferentialPrivacy.org, Aug 2020.
\newblock \url{https://differentialprivacy.org/open-problem-all-pairs/}.

\bibitem[Sei95]{seidel1995}
R.~Seidel.
\newblock On the all-pairs-shortest-path problem in unweighted undirected
  graphs.
\newblock {\em J. Comput. Syst. Sci.}, 51(3):400--403, 1995.

\bibitem[Tea17]{Apple}
Apple Differential~Privacy Team.
\newblock Learning with privacy at scale.
\newblock {\em Apple Machine Learning Journal}, 1(8), 2017.

\bibitem[Vad17]{vadhan2017textbook}
Salil Vadhan.
\newblock {\em The Complexity of Differential Privacy}, pages 347--450.
\newblock Springer International Publishing, 2017.

\bibitem[Wil18]{Williamsapsp}
R.~Ryan Williams.
\newblock Faster all-pairs shortest paths via circuit complexity.
\newblock {\em SIAM J. Comput.}, 47(5):1965--1985, 2018.

\bibitem[Zwi02]{Zwick02}
Uri Zwick.
\newblock All pairs shortest paths using bridging sets and rectangular matrix
  multiplication.
\newblock {\em J. {ACM}}, 49(3):289--317, 2002.

\end{thebibliography}
\end{document}